\documentclass[journal]{IEEEtran}
\usepackage[latin1]{inputenc}
\usepackage{amsmath}
\usepackage{graphicx}
\usepackage{amssymb}
\usepackage{xcolor}
\makeatletter

\newtheorem{lemma}{Lemma}

\newtheorem{thm}{Theorem}

\providecommand{\LyX}{L\kern-.1667em\lower.25em\hbox{Y}\kern-.125emX\@}

\def\BibTeX{{\rm B\kern-.05em{\sc i\kern-.025em b}\kern-.08em
    T\kern-.1667em\lower.7ex\hbox{E}\kern-.125emX}}

\makeatother
\begin{document}
\def\ZZ{{\mathbb Z}}
\def\RR{{\Bbb R}}
\def\NN{{\mathbb N}}
\def\CC{{\mathbb C}}

\title{Observer-Based Control of Linear Systems with Mismatched Input and Output Delays}
	\author{Hieu Trinh, Phan Thanh Nam and Tran Ngoc Nguyen
	\thanks{Hieu Trinh is with the School of Engineering, Deakin University, Waurn Ponds, 75 Pigdons Road, Geelong, Australia. (email:  hieu.trinh@deakin.edu.au)}
	\thanks{Phan Thanh Nam and Tran Ngoc Nguyen are with the Department of Mathematics, Quy Nhon University, Vietnam. (email:  phanthanhnam@qnu.edu.vn, tranngocnguyen@qnu.edu.vn)}	
}

\maketitle \maketitle \maketitle \thispagestyle{plain}
\pagestyle{plain}

\begin{abstract}

This paper investigates the stabilization of linear systems subject to simultaneous, mismatched time delays in both the control input and system output vectors. The proposed control framework is developed in two primary stages. First, an asymptotically stabilizing delayed state-feedback controller is synthesized by leveraging recent advancements in Linear Matrix Inequality (LMI) techniques. Second, this controller is realized using novel time-delay compensators \cite{trinhnam26}. This architecture successfully accommodates an output measurement delay $\tau_y$ that is independent of the input delay $\tau_u$, enabling direct estimation of the delayed state-feedback control law. The proposed methodology is then extended to target output controllers to account for simultaneous, mismatched time delays in both the control input and system output vectors.
\end{abstract}


\section{System Description and Problem Statement}

We consider the following linear system with a time delay in the control input vector
\begin{align}
	\label{a1}
	\dot{x}(t)=Ax(t)+Bu(t-\tau_u),
\end{align}
where $x(t)\in \mathbb{R}^n$ is the state vector, $u(t)\in \mathbb{R}^r$ is the control input vector, and  $\tau_u>0$ is the time delay in the input vector. Matrices
$A\in\mathbb{R}^{n\times n}$ and $B\in\mathbb{R}^{n\times r}$ are constant. Without loss of generality, let $B$ be a matrix of full column rank. 

Due to sensing or communication constraints, the measured output vector is available with a time delay and is modeled as
\begin{align}
	\label{a2}
	y(t) = Cx(t-\tau_y),
\end{align}
where  $\tau_y>0$ is the time delay, $y(t)\in\mathbb{R}^{p}$ with $0<p\le n$, and
$C\in\mathbb{R}^{p\times n}$ is a constant matrix of full row rank.

In this paper, $A$ is not necessarily Hurwitz and we consider the case where neither the current state vector $x(t)$ nor the delayed state vector $x(t-\tau_u)$ is available for feedback. Instead, we utilize the delayed output vector (\ref{a2}), in which the measurements consist of a subset of the state variables subject to a delay $\tau_y$. Our goal is to develop an observer-based control strategy that ensures the stability of the closed-loop system in the presence of distinct time delays in both the input and output channels.

This paper is organized as follows. Section II presents the main theoretical results, divided into four subsections. Section II-A details an asymptotically stabilizing state-feedback controller synthesized using advanced LMI techniques. Section II-B addresses controller implementation by introducing novel time-delay compensators capable of handling an output measurement delay $\tau_y$ independently of the input delay $\tau_u$. In Section II-C, the proposed methodology is extended to design target output controllers \cite{Fernando2025} that accommodate mismatched time delays occurring simultaneously in both the control input and system output vectors. Section II-D introduces an alternative design for reduced-order delayed functional observers. These observers are constructed from a reduced-order subsystem, which is created by projecting the full state dynamics onto the row space of the target output matrix $F_o$. The output of this resulting subsystem forms a subset of the overall output vector. Finally, concluding remarks are given in Section III.

\section{Main Results}
Our design approach comprises two main steps. In the first step, we carry out the design of a state-feedback controller of the form
\begin{align}
	\label{a3}
	u(t-\tau_u)=Fx(t-\tau_u),
\end{align}
where $F\in \mathbb{R}^{r\times n}$ is the controller gain, to be found such that the following closed-loop time-delay system is asymptotically stable
\begin{align}
	\label{a4}
	\dot{x}(t)=Ax(t)+BFx(t-\tau_u).
\end{align}

In the second step, since neither $x(t)$ nor $x(t-\tau_u)$ is available for feedback, we design a functional observer to directly estimate the delayed functional
$$z(t):=u(t-\tau_u)=Fx(t-\tau_u)$$ using the delayed output vector $y(t)$ as defined in (\ref{a2}). Finally, the designed functional is used to implement the stabilizing
state-feedback control law (\ref{a3}).

\subsection{Design of a stabilizing feedback control law $u(t-\tau_u)=Fx(t-\tau_u)$}\label{asa}
For a given $\tau_u >0$, the problem of finding $F\in \mathbb{R}^{r\times n}$ such that the closed-loop time-delay system (\ref{a4}) is asymptotically stable is solved by adapting Lemma 11 in \cite{trinhnam26}. Let us recall the following definition and fact: The set of all complex numbers \( \lambda \in \mathbb{C} \) that satisfy the characteristic equation  
\[
\det\left(\lambda I - A - A_{\tau} e^{-\lambda \tau} \right) = \det\left(\lambda I - A^{\mathsf T} - A^{\mathsf T}_{\tau} e^{-\lambda \tau} \right)=0
\]
is referred to as the \emph{spectrum} of a time-delay system $\dot{x}(t)=Ax(t)+A_{\tau}x(t-\tau)$, and is denoted by \( \sigma(A + A_{\tau} e^{-\lambda \tau}) \), where \( \sigma(A + A_{\tau} e^{-\lambda \tau})=\sigma(A^{\mathsf T} + A^{\mathsf T}_{\tau} e^{-\lambda \tau}) \). 

By transposing both $A$ and $BF$, we can apply Lemma 11 from \cite{trinhnam26} to determine $F^{\mathsf T}$. Consequently, for a given $\tau_u$, a sufficient condition for the asymptotic stability of (\ref{a4}) is the feasibility of the LMI defined in Lemma 11 \cite{trinhnam26}. It should be noted that the LMI condition of Lemma 11 \cite{trinhnam26} remains feasible for input delays up to $\bar{\tau}_u$. 
\subsection{Design of a delayed functional observer based on delayed output measurement vector}\label{asb}
Since both the current state vector $x(t)$ and the delayed state vector $x(t-\tau_u)$ are unavailable for feedback, the proposed control law (\ref{a3}) is not directly implementable. To address this, we utilize the delayed output measurement vector $y(t)=Cx(t-\tau_y)$ to construct a functional observer. This observer is designed to directly estimate the delayed functional $$z(t)=Fx(t-\tau_u),$$ bypassing the need for full state reconstruction. 

We consider two practical scenarios: Scenario 1 addresses the case where $0 < \tau_y \leq \tau_u$, while Scenario 2 deals with the case where $\tau_y > \tau_u$.

\textit{Scenario 1:} When the time delay in the output measurement vector, $\tau_y$, is within the bounded interval $0 < \tau_y \leq \tau_u$, we can define a delayed output vector $y_{\alpha}(t)$ as follows
$$y_{\alpha}(t) = y(t - \alpha)$$where $\alpha = \tau_u - \tau_y \geq 0$. In this formulation, $y_{\alpha}(t)$ is the original output $y(t)$ delayed by $\alpha$. By utilizing this new output vector to design a functional observer for estimating $z(t)$, we effectively synchronize the system; the delays present in both the input and output vectors become identical. Designing a functional observer under these conditions is significantly more straightforward, as it reduces the complexity typically introduced by mismatched delay terms. In fact, the existence conditions of such an observer are the same as those functional observers reported in the literature (see, \cite{darouach2000}). This is elaborated in the following development.

Thus, for the case where $0 < \tau_y \leq \tau_u$, to estimate $z(t)$, we consider the following observer
\begin{align}
	\label{a5}
	\hat{z}(t)&=w(t)+My_{\alpha}(t),\\
	\label{a6}
	\dot{w}(t)&=Nw(t)+Gy_{\alpha}(t)+Ju(t-2\tau_u),
\end{align} 
where $\hat{z}(t)$ is the estimate of $z(t)$. The matrices $M$, $N$, $G$ and $J$ are to be determined so that
$\hat{z}(t)\to z(t)$ asymptotically.

Defining the estimation error vector $e(t)=\hat z(t)-z(t)$, the error dynamics are given by
\begin{align}
	\label{a7}
	\dot{e}(t)&=\dot{w}(t)+M\dot{y}_{\alpha}(t)-F\dot{x}(t-\tau_u)\nonumber\\ \quad &=Ne(t)+\mathcal{C}_{1}x(t-\tau_u)+\mathcal{C}_{2}u(t-2\tau_u),
\end{align}
where\\ 
$\mathcal{C}_1 =NF+\bar{G}C+MCA-FA$, \quad $\mathcal{C}_2 =J+MCB-FB$, \quad $\bar{G}:=G-NM$, \quad $\mathcal{C}=\begin{pmatrix}
	\mathcal C_1 &
	\mathcal C_2
\end{pmatrix}$.

The following theorem characterizes the necessary and sufficient conditions for the existence of observer (\ref{a5})-(\ref{a6}).

\begin{thm}\label{thm:1p3}
	For $0 < \tau_y \leq \tau_u$, let $\alpha = \tau_u - \tau_y \geq 0$ and define the shifted output vector $y_\alpha(t) = y(t - \alpha) = Cx(t - \tau_u)$. Then observer (\ref{a5})-(\ref{a6}) of order $r=\mathrm{rank}(F)$ provides
	asymptotic estimation of the functional $z(t)=Fx(t-\tau_u)$ and yields
	estimation error dynamics that are decoupled from the plant state and input
	if and only if $\mathcal C=\bf 0$ and $N$ is Hurwitz.
	In this case, the estimation error satisfies
	\[
	e(t)=\hat z(t)-z(t)\to {\bf 0} \quad \text{as} \quad t\to\infty
	\]
	for all admissible initial conditions and inputs $u(\cdot)$.
\end{thm}

\textit{Proof:} The proof can be constructed based on the same logic as presented in the proof of Theorem 1 in \cite{trinhnam1}. Thus, it is omitted here.

\textit{Remark 1:} The conditions $\mathcal C = \mathbf{0}$ and $N$ being Hurwitz are equivalent to $J=FB-MCB$, and conditions (10) and (20) in \cite{darouach2000}. The detailed proof is omitted and left to the reader for verification.

We have demonstrated that for the case $0 < \tau_y \leq \tau_u$, the observer defined in (\ref{a5})-(\ref{a6}) effectively utilizes the output delay vector $y_{\alpha}(t)$ to estimate the functional $z(t)=Fx(t-\tau_u)$. Furthermore, our analysis confirms that the existence conditions for this observer align with those established for functional observers in \cite{darouach2000}. Consequently, standard design procedures (for example, \cite{darouach2000}, \cite{trinhfer2012}) can be employed to determine the observer gains $M$, $N$, and $G$.

Observer (\ref{a5})-(\ref{a6}) is now used to realize the state-feedback control law (\ref{a3}). Firstly, note that our designed observer ensures 	\[
e(t)=\hat z(t)-z(t)\to {\bf 0} \quad \text{as} \quad t\to\infty.
\]
By substituting \[u(t-\tau_u)=\hat{z}(t)=e(t)+Fx(t-\tau_u)\] into (\ref{a1}), we obtain
\begin{align}
	\label{a8}
	\dot{x}(t)&=Ax(t)+Bu(t-\tau_u)\nonumber\\&=Ax(t)+BFx(t-\tau_u)+Be(t).
\end{align}
The state-feedback gain $F$ is designed to stabilize (\ref{a4}) asymptotically. Likewise, the error dynamics (\ref{a7}) are asymptotically stable since $N$ is a Hurwitz matrix. 

Let us now augment (\ref{a8}) and $\dot{e}(t)=Ne(t)$ to form the following \textit{state-error} system 
\begin{align}
	\label{a9}
	\begin{pmatrix}\dot{x}(t)\\\dot{e}(t)\end{pmatrix}&=\begin{pmatrix}A &B\\ \bf 0 &N\end{pmatrix}\begin{pmatrix}x(t)\\e(t)\end{pmatrix}+\begin{pmatrix}BF &\bf 0\\ \bf 0 &\bf 0\end{pmatrix} \begin{pmatrix}x(t-\tau_u)\\e(t-\tau_u)\end{pmatrix}.
\end{align}
Equation (\ref{a9}) clearly shows that its spectrum is $\sigma(A+BF e^{-\lambda \tau_u}) \cup \sigma(N)$, which ensures the asymptotic stability of the system.

On the other hand, by substituting \[u(t-\tau_u)=\hat{z}(t)=w(t)+My_{\alpha}(t)\] into (\ref{a1}) and (\ref{a6}), we obtain the following closed-loop \textit{state-observer} system
\begin{align}
	\label{a10}
	\begin{pmatrix}\dot{x}(t)\\\dot{w}(t)\end{pmatrix}&=\begin{pmatrix}A &B\\ \bf 0  &N\end{pmatrix}\begin{pmatrix}x(t)\\w(t)\end{pmatrix}+\begin{pmatrix}BMC &\bf 0\\ GC &J\end{pmatrix} \begin{pmatrix}x(t-\tau_u)\\w(t-\tau_u)\end{pmatrix}\nonumber\\&+\begin{pmatrix}\bf 0 &\bf 0\\ JMC &\bf 0\end{pmatrix} \begin{pmatrix}x(t-2\tau_u)\\w(t-2\tau_u)\end{pmatrix}. 
\end{align}

\textit{Example 1:} Let $\tau_u=\tau_y=0.5\text{s}$, \[A=\begin{pmatrix} 0  &   1\\-1  & 1 \end{pmatrix}, \quad B=\begin{pmatrix} 0\\1 \end{pmatrix}, \quad C=\begin{pmatrix} 1 &0 \end{pmatrix}. \]
Our objective is to design an observer-based controller to stabilize the system when both the input and output vectors are subject to a time delay of $0.5\text{s}$. Note that $A$ is not Hurwitz as $\sigma(A)=\{0.5\pm j0.866\}.$

Firstly, by using Lemma 11 in \cite{trinhnam26}, we find $F\in\mathbb{R}^{1\times 2}$ such that the closed-loop system (\ref{a4}) is asymptotically stable.

By transposing both $A$ and $BF$ and adopting the notation from Lemma 11 \cite{trinhnam26}, we obtain
\[A^{\mathsf T}:=N_{0,1}+ZN_{0,2}=\begin{pmatrix} 0  &   -1\\1  & 1 \end{pmatrix},\] 
\[F^{\mathsf T}B^{\mathsf T}:=N_{\tau,1}+ZN_{\tau,2}=Z\begin{pmatrix} 0  &1 \end{pmatrix},\]
where $N_{0,1}=\begin{pmatrix} 0  &   -1\\1  & 1 \end{pmatrix}$, $N_{0,2}=\begin{pmatrix} 0  &0 \end{pmatrix}$, $N_{\tau,1}=\textbf{ 0}_{2\times 2}$, $F^{\mathsf T}=Z$, and $B^{\mathsf T}=N_{\tau,2}=\begin{pmatrix} 0  &1 \end{pmatrix}$.

With $\tau_u=0.5$, $N_{0,1}$, $N_{0,2}$, $N_{\tau,1}$ and $N_{\tau,2}$ as given above, the LMI in Lemma 11 \cite{trinhnam26} is feasible for $\lambda=1$, and we obtain
\[Z=F^{\mathsf T}=\begin{pmatrix} 0.6979\\
	-1.2628 \end{pmatrix}.\]
Consequently, the state-feedback control law 
\[u(t-0.5)=Fx(t-0.5)=\begin{pmatrix} 0.6979 &-1.2628 \end{pmatrix}x(t-0.5),\]
guarantees the asymptotic stability of the following closed-loop time-delay system
\[\dot{x}(t)=\begin{pmatrix} 0  &   1\\-1  & 1 \end{pmatrix}x(t)+\begin{pmatrix} 0  &   0\\0.6979  & -1.2628 \end{pmatrix}x(t-0.5).\]

Following the methodology in \cite{wu}, those stable eigenvalues nearest to the imaginary axis are found to be $\{-0.5664 \pm j1.4192, -0.59\}$. Furthermore, the LMI condition presented in Lemma 11 \cite{trinhnam26} remains feasible for input delays up to $\bar{\tau}_u = 0.95\text{s}$. 

Next, we design a functional observer to estimate the delayed functional
$$z(t)=\begin{pmatrix} 0.6979 & -1.2628 \end{pmatrix}x(t-0.5).$$
This estimation utilizes the delayed output measurement $y(t)=x_1(t-0.5)$. Since the conditions defined in (10) and (20) of \cite{darouach2000} are satisfied, we have the flexibility to assign the pole of $N$ anywhere within the complex $s$-plane. By selecting $N=-4$, we obtain the following functional observer
\[\hat{z}(t)=w(t)-5.6163y_{\alpha}(t),\]
\[\dot{w}(t)=-4w(t)+26.5196y_{\alpha}(t)-1.2628u(t-1).\]

The combination of the designed controller and observer yields an asymptotically stable closed-loop state-observer system (\ref{a10}). Figure \ref{fig1paper3} shows the trajectories of $x_1(t)$ and  $x_2(t)$. It is clear
that asymptotic stability of the closed-loop system has been achieved.
\begin{figure}[!h]
	\centering
	\includegraphics[width=0.9\linewidth]{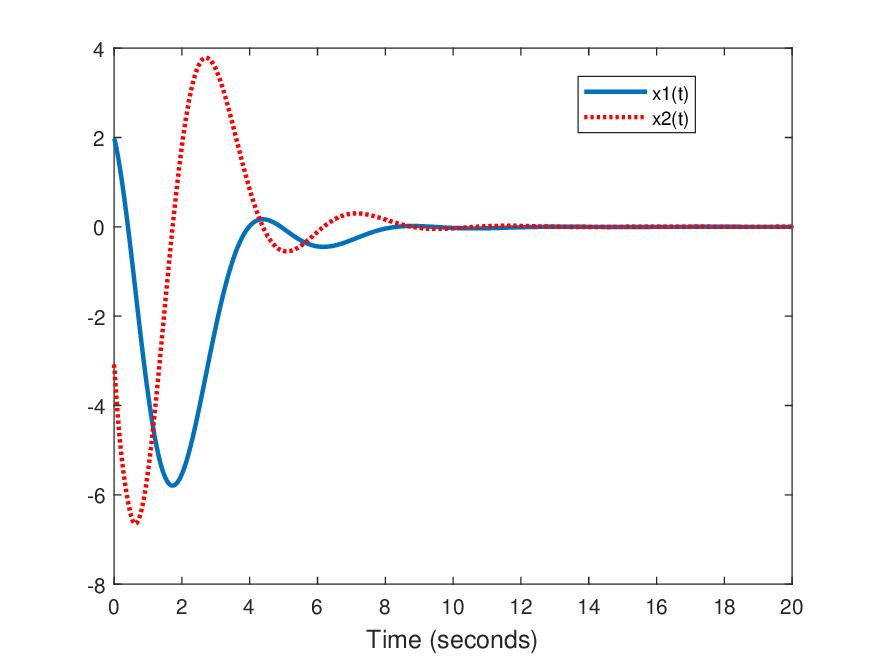}
	\caption{Trajectories of $x_1(t)$ and $x_2(t)$}
	\label{fig1paper3}
\end{figure}

\textit{Scenario 2:} We consider the case where $\tau_y > \tau_u$. 

Given the mismatched delay conditions, our objective is to design an asymptotic observer that directly estimates the delayed functional $z(t)=Fx(t-\tau_u)$ using the delayed output vector $y(t)=Cx(t-\tau_y)$. 
Building on established techniques for instantaneous functional estimation \cite{trinhnam26, trinhnam1}, our method focuses on the direct reconstruction of the \textit{delayed} state functional.

Let us look at the following observer which incorporates an internal delay dynamics and it is driven
by the delayed output measurements and input signals. 
\begin{align}
	\label{a11}
	\hat{z}(t)&= w(t)+My(t-\tau_u),\\
	\label{a12}
	\dot{w}(t)&=Nw(t)+N_{\tau_y}w(t-\tau_y)+Gy(t-\tau_u)  \nonumber\\
	&+G_1y(t-\tau_u-\tau_y)+ Ju(t-2\tau_u)	\nonumber\\
	&+ J_1u(t-2\tau_u-\tau_y), 
\end{align}
with $w(\theta)=\rho(\theta)$  for $\theta\in[-\tau_y,0]$ where $w(t)\in\mathbb{R}^r$.
The matrices $M$, $N$, $N_{\tau_y}$, $G$, $G_1$, $J$ and $J_1$ are to be determined so that
$\hat{z}(t)\to z(t)$ asymptotically.

Defining the estimation error vector $e(t)=\hat z(t)-z(t)$, the error dynamics are given by
\begin{align}
	\label{a13}
	\dot{e}(t)&=\dot{w}(t)+M\dot{y}(t-\tau_u)-F\dot{x}(t-\tau_u)\nonumber\\ &=Ne(t)+N_{\tau_y}e(t-\tau_y)+\mathcal{\bar{C}}_{1}x(t-\tau_u)\nonumber\\ &+\mathcal{\bar{C}}_{2}x(t-\tau_u-\tau_y)+\mathcal{\bar{C}}_{3}x(t-\tau_u-2\tau_y)\nonumber\\ &+\mathcal{\bar{C}}_{4}u(t-2\tau_u)+\mathcal{\bar{C}}_{5}u(t-2\tau_u-\tau_y),
\end{align}
where\\ 
$\mathcal{\bar{C}}_1 =NF-FA$, \ $\mathcal{\bar{C}}_2 =N_{\tau_y}F+\bar{G}C+MCA$, \ $\mathcal{\bar{C}}_3 =\bar{G}_1C$, \ $\mathcal{\bar{C}}_4 =J-FB$, \ $\mathcal{\bar{C}}_5 =J_1+MCB$, \ $\bar{G}:=G-NM$, $\bar{G}_1:=G_1-N_{\tau_y}M$, and $\mathcal{\bar{C}}=\begin{pmatrix}
	\mathcal {\bar{C}}_1 &
	\mathcal {\bar{C}}_2 & \mathcal {\bar{C}}_3 &\mathcal {\bar{C}}_4 &\mathcal {\bar{C}}_5
\end{pmatrix}$.

\textit{Remark 2:} Based on the error dynamics equation presented above, the existence conditions are established. These conditions require that $\bar{\mathcal{C}} = \mathbf{0}$ and that the system described by the following delayed differential equation is asymptotically stable \begin{align}
	\label{a14}\dot{e}(t) = Ne(t) + N_{\tau_y}e(t - \tau_y).
\end{align}

Notably, the existence conditions for the observer in (\ref{a11})-(\ref{a12}) are identical to those established in Theorem 1 of \cite{trinhnam26}. As a result, the systematic design procedure detailed in \cite{trinhnam26} can be directly applied to calculate the observer gains $M$, $N$, $N_{\tau_y}$, $G$, $G_1$, $J$ and $J_1$.

Let us pause here for an illustrative example.

\textit{Example 2:} Reconsidering Example 1, recall that the LMI condition from Lemma 11 \cite{trinhnam26} is feasible for input delays up to $\bar{\tau}_u = 0.95$s. We now examine a scenario where the output delay $\tau_y = 1.1$s exceeds this threshold ($\tau_y > \bar{\tau}_u$). In this case, our objective is to design a functional observer to estimate the following functional from Example 1
$$z(t) = \begin{pmatrix} 0.6979 & -1.2628 \end{pmatrix} x(t - 0.5)$$
using the delayed output measurement
\[y(t)=Cx(t-1.1)=x_1(t-1.1).\]
This estimation is performed to facilitate the implementation of an observer-based closed-loop control system.

Given $F=\begin{pmatrix} 0.6979 & -1.2628 \end{pmatrix}$ and $A = \begin{pmatrix} 0 & 1 \\ -1 & 1 \end{pmatrix}$, the standard rank condition is not satisfied, as $\text{rank} \begin{pmatrix} FA \\ F \end{pmatrix} > \text{rank}(F)$. 

\textit{Remark 3:} Following the methodology in \cite{trinhnam26}, we increase the order of the observer by augmenting the original functional $z(t)=Fx(t-\tau_u)$ with an auxiliary functional $z_a(t)=Rx(t-\tau_u)$. The resulting augmented functional is defined as
$$z_{\text{aug}}(t)=\begin{pmatrix} z(t)\\z_a(t)\end{pmatrix}=\begin{pmatrix} F\\R\end{pmatrix} x(t-\tau_u)=\bar{F}x(t-\tau_u),$$where $\bar{F}:=\begin{pmatrix} F\\R \end{pmatrix}\in\mathbb{R}^{q \times n}$. The matrix $R \in \mathbb{R}^{(q-r) \times n}$ is systematically constructed (per \cite{trinhnam1}) to satisfy the rank condition
$$\text{rank} \begin{pmatrix} \bar{F}A \\ \bar{F} \end{pmatrix} = \text{rank}(\bar{F}).$$By designing a functional observer of order $q$ to estimate $z_{\text{aug}}(t)$, the original estimate $\hat{z}(t)$ can be extracted via
$$\hat{z}(t) = \begin{pmatrix} I_r & \mathbf{0}_{r \times (q-r)} \end{pmatrix} \hat{z}_{\text{aug}}(t) = K\hat{z}_{\text{aug}}(t),$$where $K := \begin{pmatrix} I_r & \mathbf{0}_{r \times (q-r)} \end{pmatrix}$.

For this specific example, choosing $R = \begin{pmatrix} 0 & 1 \end{pmatrix}$ ensures that the rank condition is satisfied. With $\bar{F}=\begin{pmatrix} F\\R \end{pmatrix}=\begin{pmatrix}0.6979 & -1.2628\\0 &1 \end{pmatrix}$, we obtain \[N=\bar{F}A\bar{F}^-=\begin{pmatrix}1.8095  &  1.7202\\
	-1.4329  & -0.8095 \end{pmatrix}.\]

In this case, matrix $\begin{pmatrix} C \\ CA \end{pmatrix}$ has a rank of 2. Since it is full rank, we can follow the approach detailed in Case 1 \cite{trinhnam26}. This involves first addressing the stabilization problem by determining a gain $N_{\tau}$ that ensures the time-delay error system (\ref{a14}) is asymptotically stable.
Numerical analysis reveals that the LMI condition proposed in Lemma 11 \cite{trinhnam26} maintains feasibility for output delays up to $\bar{\tau}_y = 1.26\text{s}$, a value significantly larger than $\bar{\tau}_u$. 

Now, as for $\tau_y = 1.1\text{s}$, the LMI condition in Lemma 11 \cite{trinhnam26} is feasible, and we obtain 
\[N_{\tau_y}=\begin{pmatrix} -1.2014 &-1.0186\\0.8709 &0.3318 \end{pmatrix}.
\]
Based on the design procedure as presented in Case 1 \cite{trinhnam26}, we obtain the following second-order observer
\begin{align}
	&\hat{z}(t)=w_1(t)
	-0.4986y(t-0.5),\nonumber\\ &\dot{w}(t)=\begin{pmatrix}1.8095  &  1.7202\\
		-1.4329  & -0.8095 \end{pmatrix}w(t)+\begin{pmatrix}1.2573\\-0.5150 \end{pmatrix}y(t-0.5) \nonumber\\&+ \begin{pmatrix} -1.2014 &-1.0186\\0.8709 &0.3318 \end{pmatrix}w(t-1.1)+\begin{pmatrix}-1.2628\\1 \end{pmatrix}u(t-1)\nonumber\\&+\begin{pmatrix}-0.1832\\-0.1794 \end{pmatrix}y(t-1.6).\nonumber
\end{align}
The combination of the designed controller and observer yields an asymptotically stable closed-loop system. Figure \ref{fig2paper3} shows the trajectories of $x_1(t)$ and  $x_2(t)$. It is clear
that asymptotic stability of the closed-loop system has been achieved.
\begin{figure}[!h]
	\centering
	\includegraphics[width=0.9\linewidth]{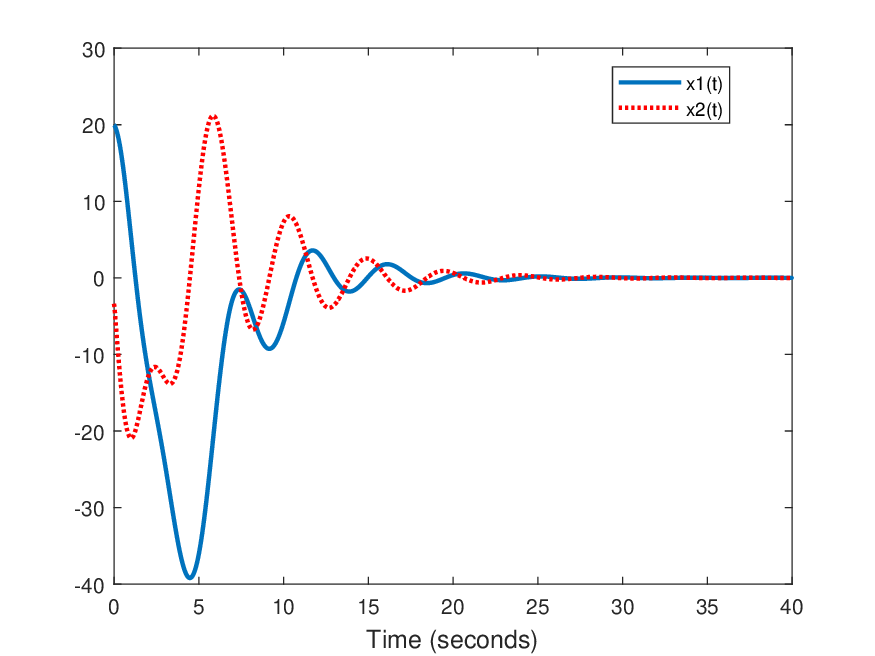}
	\caption{Trajectories of $x_1(t)$ and $x_2(t)$}
	\label{fig2paper3}
\end{figure}

\textit{Remark 4:} While the LMI condition in Lemma 11 \cite{trinhnam26} is technically feasible for output delays up to $\bar{\tau}_y$, the observer structure in \eqref{a11}-\eqref{a12} fails to remain effective if $\bar{\tau}_y \leq \bar{\tau}_u$. To address this, the permissible output delay $\tau_y$ must be expanded. Following the approach in \cite{trinhnam26}, this is accomplished by using an augmented measurement vector within a generalized observer framework that utilizes multiple internal delay channels.

Alternatively, under the stricter assumption that the full delayed state vector is measurable (i.e., $\text{rank}(C) = n$), a functional observer can be designed to estimate
$z(t) = Fx(t - \tau_u)$. This approach remains viable for significantly larger output delays where $\tau_y > \tau_u$. In this configuration, the observer acts as a time-delay compensator, leveraging heavily delayed state measurements to reconstruct a more current (less-delayed) state estimate.

Accordingly, assuming the available delayed output vector is $y(t) = Cx(t - \tau_y)$ with $\text{rank}(C) = n$, we propose the following observer. This architecture incorporates multiple internal delay dynamics and is driven by both delayed output measurements and control input signals
\begin{align}
	\label{a15}
	\dot{\hat{z}}(t)&=N\hat{z}(t)+N_{\tau_u}\hat{z}(t-\tau_u)+N_{\tau_y}\hat{z}(t-\tau_y)
	+G_{\alpha} y(t-\alpha)  \nonumber\\
	&+G_{\tau_u}y(t-\tau_u)+FBu(t-2\tau_u), 
\end{align}
with $\hat{z}(\theta)=\rho(\theta)$  for $\theta\in[-\tau_y,0]$, and $\alpha > 0$ is a parameter to be specified.
The matrices $N$, $N_{\tau_u}$, $N_{\tau_y}$, $G_{\alpha}$ and $G_{\tau_u}$ are to be determined so that
$\hat{z}(t)\to z(t)$ asymptotically. 

Defining the estimation error vector $e(t)=\hat z(t)-Fx(t-\tau_u)$, the error dynamics are given by
\begin{align}
	\label{a16}
	\dot{e}(t)&=\dot{\hat{z}}(t)-F\dot{x}(t-\tau_u)=Ne(t)+N_{\tau_u}e(t-\tau_u)\nonumber\\ \quad &+N_{\tau_y}e(t-\tau_y)+\mathcal{\bar{C}}_{1}x(t-\tau_u)+\mathcal{\tilde{C}}_{1}x(t-\tau_u-\tau_y)\nonumber\\&+N_{\tau_u}Fx(t-2\tau_u)+G_{\alpha}Cx(t-\alpha-\tau_y),
\end{align}
where $\mathcal{\bar{C}}_1 =NF-FA$ and $\mathcal{\tilde{C}}_1 =N_{\tau_y}F+G_{\tau_u}C$.

To simplify (\ref{a16}), we select the parameter $\alpha \geq 0$ to satisfy the following delay matching condition
\begin{align}\label{a17}\alpha + \tau_y = 2\tau_u.\end{align}Applying this condition allows us to consolidate the final two terms of the error dynamics into a single expression, i.e.,
$$N_{\tau_u}Fx(t-2\tau_u) + G_{\alpha}Cx(t-\alpha-\tau_y) = \mathcal{\tilde{C}}_{2}x(t-2\tau_u),$$where $\mathcal{\tilde{C}}_{2} := N_{\tau_u}F + G_{\alpha}C$.

Note that for a valid solution $\alpha \geq 0$ to exist in (\ref{a17}), the measurement delay $\tau_y$ must satisfy the upper bound
$$\tau_y \leq 2\tau_u.$$

The following theorem characterizes sufficient conditions for the existence of observer (\ref{a15}).
\begin{thm}\label{thm:2p3}
	For $\tau_u<\tau_y\leq 2\tau_u$, observer (\ref{a15}) provides
	asymptotic estimation of the functional $z(t)=Fx(t-\tau_u)$ 
	if $\begin{pmatrix} \mathcal {\bar{C}} & \mathcal{\tilde{C}}_{1} &\mathcal{\tilde{C}}_{2} \end{pmatrix}=\bf 0$ and the following delay-dependent error
	dynamics
	\begin{align}
		\label{a18}&\dot{e}(t)=Ne(t)+N_{\tau_u}e(t-\tau_u)+N_{\tau_y}e(t-\tau_y)
	\end{align}
	is asymptotically stable. 
\end{thm}

\textit{Proof:} If $\begin{pmatrix} \mathcal {\bar{C}} & \mathcal{\tilde{C}}_{1} &\mathcal{\tilde{C}}_{2} \end{pmatrix}=\bf 0$, then \eqref{a16} reduces to \eqref{a18}, so the error dynamics are
decoupled from $x(\cdot)$ and $u(\cdot)$. If, in addition, \eqref{a18} is asymptotically stable, then $e(t)\to \bf 0$ as $t\to\infty$ for all admissible initial conditions and inputs $u(\cdot)$. This completes the proof.

According to \cite{trinhnam26}, the constraint $\mathcal{\bar{C}}_{1}=\mathbf{0}$ is satisfied if and only if
$$\text{rank} \begin{pmatrix} FA \\ F \end{pmatrix} = \text{rank}(F).$$Under this specific condition, the matrix $N$ is defined as $N=FAF^-$.

Given that the matrix $C$ has full rank ($\text{rank}(C)=n$), the constraints $\mathcal{\tilde{C}}_{1}=\mathcal{\tilde{C}}_{2}=\mathbf{0}$ are satisfied if the following conditions hold
\begin{align}
	\label{a19}
	G_{\tau_u}&=-N_{\tau_y}FC^{-},\\
	\label{a20}
	G_{\alpha}&=-N_{\tau_u}FC^{-}.
\end{align}

Thus, for given values of $\tau_u$ and $\tau_y$ satisfying the condition $\tau_u < \tau_y \leq 2\tau_u$, Lemma 13 in \cite{trinhnam26} can be employed to determine $N_{\tau_u}$ and $N_{\tau_y}$ such that system (\ref{a18}) is asymptotically stable. A sufficient condition for this stability is the feasibility of the LMI presented in Lemma 13 of \cite{trinhnam26}.

We now revisit Example 2 to design an observer-based controller that stabilizes the system for the case where $\tau_u = 0.75\text{s}$, $\tau_y = 1.45\text{s}$, and $C = I_2$. Under these conditions, $C = I_2$ allows us to utilize the observer (\ref{a15}) to estimate the control law$$z(t) = u(t - 0.75) = Fx(t - 0.75)$$where the gain matrix$$F = \begin{pmatrix} 0.8121 & -0.8469 \end{pmatrix}$$is obtained by applying Lemma 11 \cite{trinhnam26} (with $\tau_u = 0.75\text{s}$ and $\lambda = 1$).

Note that while the input time-delay here is significantly larger than the $0.5\text{s}$ delay in Example 2, it remains below the threshold $\bar{\tau}_u = 0.95\text{s}$. Consequently, the feasibility of the LMI in Lemma 11 \cite{trinhnam26} is still guaranteed.

Conversely, because the measurement delay exceeds the input delay ($\tau_y = 1.45\text{s} > 0.75\text{s}$), we cannot directly obtain $z(t)=Fx(t-0.75)$ from the heavily delayed state measurement vector $x(t-1.45)$, despite having $C=I_2$. Furthermore, since $\tau_y = 1.45\text{s}$ exceeds $\bar{\tau}_y=1.26\text{s}$, the observer (\ref{a11})-(\ref{a12}) cannot be used to implement the control law. To overcome these limitations, we utilize the more general observer (\ref{a15}) as a time-delay compensator. This structure leverages the heavily delayed measurements to reconstruct a more recent, less-delayed functional estimate.

With $\tau_u=0.75\text{s}$ and $\tau_y=1.45\text{s}$ satisfying $\tau_y < 2\tau_u$, it follows that $\alpha >0$. Evaluating equation (\ref{a17}) yields $\alpha = 0.05$. 

As in Example 2, the rank condition is violated for $F = \begin{pmatrix} 0.8121 & -0.8469 \end{pmatrix}$ and $A = \begin{pmatrix} 0 & 1 \\ -1 & 1 \end{pmatrix}$ since $\text{rank} \begin{pmatrix} FA \\ F \end{pmatrix} > \text{rank}(F)$. Therefore, we insert an additional row by choosing $R = \begin{pmatrix} 0 & 1 \end{pmatrix}$ to satisfy the rank condition. With $\bar{F}=\begin{pmatrix} F\\R \end{pmatrix}=\begin{pmatrix}0.8121 & -0.8469\\0 &1 \end{pmatrix}$, we obtain \[N=\bar{F}A\bar{F}^-=\begin{pmatrix}1.0429 &   0.8484\\	-1.2314  & -0.0429 \end{pmatrix}.\]
Now, with $\tau_u=0.75\text{s}$ and $\tau_y=1.45\text{s}$ the LMI in Lemma 13  \cite{trinhnam26} is feasible for $\lambda = 1$, and we obtain $N_{\tau_u}$ and $N_{\tau_y}$ which ensure the asymptotic stability of (\ref{a18}), where
\[N_{\tau_u}=\begin{pmatrix} -0.8705 &-0.3991\\0.9424 &  -0.4705\end{pmatrix}, \] \[N_{\tau_y}=\begin{pmatrix} -0.0650 &-0.0002\\-0.1714 &  -0.0400 \end{pmatrix}.
\]
Figure \ref{fig3paper3} shows the trajectories $e_1(t)$ and  $e_2(t)$ of the error time-delay system (\ref{a18}), where it is clear that both converge asymptotically to zero.
\begin{figure}[!h]
	\centering
	\includegraphics[width=0.9\linewidth]{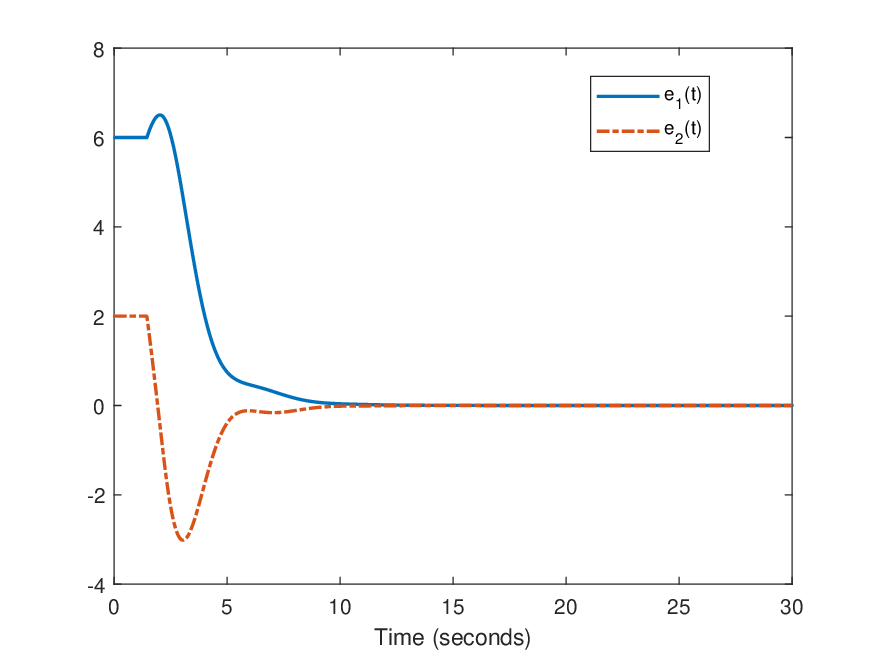}
	\caption{Trajectories of $e_1(t)$ and $e_2(t)$}
	\label{fig3paper3}
\end{figure}

We obtain $\bar{F}B=\begin{pmatrix} -0.8469\\
		1\end{pmatrix}$, and from (\ref{a19})-(\ref{a20}), the observer gains
\[G_{\tau_u}=\begin{pmatrix} 0.0528 &  -0.0548\\0.1392 &   -0.1051\end{pmatrix},  G_{\alpha}=\begin{pmatrix} 0.7070 &-0.3382\\-0.7654 &   1.2687 \end{pmatrix}.
\]
Finally, the functional $z(t)=u(t-0.75)=Fx(t-0.75)$ can now be estimated according to the following equations
$$\hat{z}(t) = \begin{pmatrix} 1 & 0\end{pmatrix} \hat{z}_{\text{aug}}(t),$$
\begin{align}
	\dot{\hat{z}}_{\text{aug}}(t)&=\begin{pmatrix}1.0429 &   0.8484\\	-1.2314  & -0.0429 \end{pmatrix}\hat{z}_{\text{aug}}(t)\nonumber\\&+\begin{pmatrix} -0.8705 &-0.3991\\0.9424 &  -0.4705\end{pmatrix}\hat{z}_{\text{aug}}(t-0.75) \nonumber\\&+ \begin{pmatrix} -0.0650 &-0.0002\\-0.1714 &  -0.0400 \end{pmatrix}\hat{z}_{\text{aug}}(t-1.45)\nonumber\\&+\begin{pmatrix} 0.7070 &-0.3382\\-0.7654 &   1.2687 \end{pmatrix}y(t-0.05)\nonumber\\&+\begin{pmatrix} 0.0528 &  -0.0548\\0.1392 &   -0.1051\end{pmatrix}y(t-0.75)\nonumber\\&+\begin{pmatrix} -0.8469\\
		1\end{pmatrix}u(t-1.5).\nonumber
\end{align}
To illustrate the validity of the observer-based control scheme which combines the controller and observer derived above, Figure \ref{fig3apaper3} shows the trajectories of $x_1(t)$ and  $x_2(t)$. It is clear
that asymptotic stability of the closed-loop system has been achieved.
\begin{figure}[!h]
	\centering
	\includegraphics[width=0.9\linewidth]{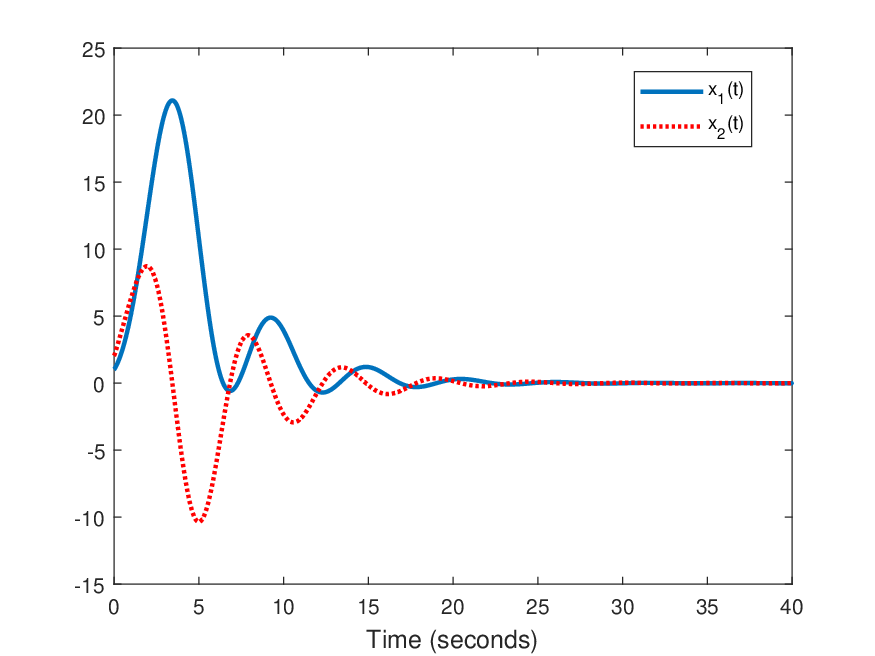}
	\caption{Trajectories of $x_1(t)$ and $x_2(t)$}
	\label{fig3apaper3}
\end{figure}

\textit{Remark 5:} Notably, the LMI in Lemma 13 \cite{trinhnam26} remains feasible for values of $\tau_y$ significantly exceeding $1.45\text{s}$, even when $\tau_u > \bar{\tau}_u$. However, given the constraint (\ref{a17}) (where $\tau_y \leq 2\tau_u$), observer (\ref{a15}) is most effective for larger delays satisfying $\tau_y \geq \bar{\tau}_y$, and $\tau_u$ is closer to the limit $ \bar{\tau}_u$ as illustrated by the case where $\tau_u = 0.75\text{s}$ and $\tau_y = 1.45\text{s}$.

Conversely, because the LMI condition in Lemma 11 \cite{trinhnam26} is restricted to input delays within $\bar{\tau}_u$, it cannot provide a stabilizing controller gain for any $\tau_u > \bar{\tau}_u$. To increase the allowable upper bound $\bar{\tau}_u$, one may utilize a state feedback control law incorporating multiple delays, such as
\begin{align}\label{a21}u(t-\tau_u)=Fx(t-\tau_u)+F_hx(t-h),\end{align}
where $h>\tau_u$ and $F_h\in \mathbb{R}^{r\times n}$. Applying the control law (\ref{a21}) to the system (\ref{a1}) yields the following closed-loop time-delay system
\begin{align}\label{a22}\dot{x}(t)=Ax(t)+BFx(t-\tau_u) + BF_hx(t-h).\end{align}
The objective is to ensure this closed-loop system is asymptotically stable. As demonstrated in Section V of \cite{trinhnam26}, system (\ref{a22}) maintains stability under a larger delay $\tau_u$ compared to (\ref{a4}). By applying Lemma 13 from the same work, the gain matrices $F$ and $F_h$ can be determined; specifically, the feasibility of the LMI in Lemma 13 \cite{trinhnam26} provides a sufficient condition to ensure the asymptotic stability of (\ref{a22}).

Finally, the control law (\ref{a21}) is implemented by designing a generalized functional observer that utilizes the delayed measurement vector $y(t)=Cx(t-\tau_y)$ to directly estimate the delayed generalized functional
$$z(t)=Fx(t-\tau_u)+F_hx(t-h).$$The reader is encouraged to explore this further; in particular, the design methodologies recently reported in \cite{trinhvan26}, which discuss generalized functional observers, can be adopted to address this problem.

\subsection{Target output controller design for linear systems with state and output delays}\label{asc}
In many control applications, regulating or tracking specific linear combinations of system states or outputs--referred to as target outputs--is often preferred over controlling the entire state vector. Motivated by recent findings in \cite{Fernando2025}, this section investigates the design of target output controllers for linear systems subject to time delays in both the input and output vectors. Specifically, for system (\ref{a1}), we  consider the design and implementation of the target output controller\begin{equation}\label{a23}u(t-\tau_u)=Z_{\tau_u}z_o(t-\tau_u),\end{equation}
where $z_o(t)=F_ox(t)\in \mathbb{R}^m$ represents the target output vector, $F_o\in \mathbb{R}^{m\times n}$ is full row rank, and $Z_{\tau_u}\in \mathbb{R}^{r\times m}$ denotes the controller gain.

The primary objective of control law (\ref{a23}) is to drive the target output vector to the origin ($z_o(t) \to \mathbf{0}$ as $t \to \infty$) from any initial condition $F_ox(0)$. The design of the controller is based on a reduced-order
system \cite{Fernando2025} obtained by projecting the full state dynamics onto the
row space of the target output matrix $F_o$.

In practice, however, neither $z_o(t)$ nor its delayed counterpart $z_o(t-\tau_u)$ is directly available for feedback. Additionally, the measured output vector is subject to the time delay modeled in (\ref{a2}). To overcome these limitations, the controller (\ref{a23}) is implemented using a functional observer designed to directly estimate either the delayed target output vector $z_o(t-\tau_u)$ or the delayed control input vector $u(t-\tau_u)$ (the exact formulation is detailed later). By applying this design methodology, this work extends the results of \cite{Fernando2025} to accommodate mismatched time delays occurring simultaneously in both the control input and system output vectors.

We will use the following lemmas \cite{Fernando2025} in the sequel.
\begin{lemma}[\cite{Fernando2025}]\label{lemma1p3}
	The following statements are equivalent:
	\begin{itemize}
		\item[i)] $\mathrm{rank}\begin{pmatrix}
			F_oA\\F_o
		\end{pmatrix} =\mathrm{rank}(F_o)$,
		\item[ii)] $	F_oA(I-F^{-}_oF_o) = \mathbf{0}.$  
	\end{itemize}
\end{lemma}
\begin{lemma}[\cite{Fernando2025}]\label{lemma2p3}
	The following equation
	\begin{equation} 
		N_oF_o-F_oA=\bf 0, \nonumber 
	\end{equation} 
	where $A\in \mathbb{R}^{n\times n}$, $F_o\in \mathbb{R}^{m\times n}$, $\mathrm{rank}(F_o)=m$ and $m \leq n$ are known matrices and 
	$N_o \in \mathbb{R}^{m\times m}$ is an unknown matrix, has a solution
	if and only if
	\begin{equation} 
		\mathrm{rank}\begin{pmatrix}
			F_oA\\F_o
		\end{pmatrix} =\mathrm{rank}(F_o), \nonumber 
	\end{equation}
	and in this situation $N_o=F_oAF^-_o$ satisfies
	\begin{equation}
		\sigma(N_o) \subseteq \sigma(A). \nonumber 
	\end{equation}
\end{lemma}

Following the methodology presented by Fernando and Darouach \cite{Fernando2025}, we derive the existence conditions for controller (23) by focusing on a subsystem of order $m$, where $m$ represents the number of target functionals to be controlled.

Pre-multiplying \eqref{a1} by $F_o$, we obtain
\begin{equation} 
	F_o\dot{x}(t) = F_oAx(t) + F_oBu(t-\tau_u) \nonumber
\end{equation}	
which can be expressed as
\begin{align}
	&\dot{z}_o(t)= F_oA(I-F^-_oF_o+F^-_oF_o)x(t) +F_oBu(t-\tau_u) \nonumber\\
	&=N_oz_o(t) + F_oA(I-F^-_oF_o)x(t)+B_ou(t-\tau_u), \label{a24}
\end{align}
where $B_o:=F_oB$ and $N_o:=F_oAF^-_o$.
\begin{thm}\label{thm:3p3}
	The control law $u(t-\tau_u)=Z_{\tau_u}z_o(t-\tau_u), Z_{\tau_u}\in \mathbb{R}^{r\times m}$, can drive $z_o(t) \rightarrow \bf{0}$ as $t \rightarrow \infty$ from any given initial 
	condition $F_ox(0)$
	by stabilizing a time-delay system of order $m$ if and only if the following conditions are satisfied,
	\begin{align}
		\label{a25}
		&\mathrm{rank}\begin{pmatrix}
			F_oA\\F_o
		\end{pmatrix} =\mathrm{rank}(F_o),\\
		&
		\sigma\Big(N_o + B_oZ_{\tau_u}e^{-\lambda \tau_u}\Big) \subset \mathbb{C}_- \label{a26}
	\end{align}
	where $\mathbb{C}_- := \{ \lambda \in \mathbb{C} \mid \operatorname{Re}(\lambda) < 0\}$.
\end{thm}

\textit{Proof:} From Lemma \ref{lemma1p3}, condition \eqref{a25} is equivalent to
\[F_oA(I-F^-_oF_o)=\mathbf{0}.\] We first prove sufficiency. If \eqref{a25} holds, then from the projected dynamics in \eqref{a24}, we obtain the following reduced-order subsystem of dimension $m$
\begin{align}
	\label{a27}
	\dot{z}_o(t)= N_oz_o(t) +B_ou(t-\tau_u).
\end{align}
Substituting the control law $u(t-\tau_u)=Z_{\tau_u}z_o(t-\tau_u)$ into ({\ref{a27}) gives the closed-loop time-delay system
	\begin{align}
		\label{a28}
		\dot{z}_o(t)= N_oz_o(t) +B_oZ_{\tau_u}z_o(t-\tau_u).
	\end{align}
	If condition \eqref{a26} holds, then the time-delay system (\ref{a28}) is asymptotically stable. Hence, $z_o(t) \rightarrow \mathbf{0}$ as $t \rightarrow \infty$, and the target output is stabilized via a subsystem of order $m$.
	
	Following Fernando and Darouach \cite{Fernando2025}, we establish necessity by proving the contrapositive. Suppose that \eqref{a25} holds, but condition \eqref{a26} is violated. It follows that the time-delay system \eqref{a28} possesses at least one eigenvalue in the closed right half of the complex plane. Consequently, the system \eqref{a28} is not asymptotically stable, which implies that $z_o(t) \not\to \mathbf{0}$ as $t \to \infty$. This completes the proof of necessity, thereby concluding the proof of Theorem \ref{thm:3p3}.
	
	As detailed in Section \ref{asa}, the objective is to determine $Z_{\tau_u}$ such that the asymptotic stability of \eqref{a28} is guaranteed. For a given time-delay $\tau_u$, a sufficient condition for the asymptotic stability of \eqref{a28} is the feasibility of the LMI defined in Lemma 11 \cite{trinhnam26}. Hence, $Z_{\tau_u}$ is obtained which ensures that $\sigma\Big(N_o + B_oZ_{\tau_u}e^{-\lambda \tau_u}\Big) \subset \mathbb{C}_-$.
	
	Applying the control law (\ref{a23}) to the full-order system (\ref{a1}) yields the closed-loop dynamics$$\dot{x}(t)=Ax(t)+BZ_{\tau_u}F_ox(t-\tau_u),$$whose spectrum is given by$$\sigma\Big(A + BZ_{\tau_u}F_oe^{-\lambda \tau_u}\Big).$$
	
	Subject to the satisfaction of (\ref{a25}), the matrix equation $N_oF_o-F_oA=\bf 0$ can be rewritten as follows
	\[(N_o+B_oZ_{\tau_u}e^{-\lambda \tau_u})F_o-F_o(A + BZ_{\tau_u}F_oe^{-\lambda \tau_u})=\bf 0.\]
	Since $F_o\neq \bf 0$ is full row rank, from Lemma \ref{lemma2p3}, it follows that \[
	\sigma\left( N_o+B_oZ_{\tau_u}e^{-\lambda \tau_u} \right)
	\subseteq
	\sigma\left( A + BZ_{\tau_u}F_oe^{-\lambda \tau_u} \right).
	\]

By demonstrating that the spectral properties of the reduced-order closed-loop system are embedded within the full-order closed-loop system, we establish a direct mapping of their stability characteristics. This justifies using the reduced-order model (\ref{a27}) as a valid framework for designing the target output controller.

As noted in Remark 3, when condition \eqref{a25} is violated, one can construct an augmented matrix $\bar{F}_o$ of the form
$$ \bar{F}_o = \begin{pmatrix} F_o \\ R \end{pmatrix} $$by choosing a matrix $R$ such that substituting $\bar{F}_o$ for $F_o$ in \eqref{a25} satisfies the condition. The modified controller then regulates the augmented target output vector
\[\bar{z}_o(t)=\begin{pmatrix} z_o(t)\\z_{a}(t)\end{pmatrix}=\begin{pmatrix} F_o\\R\end{pmatrix} x(t)=\bar{F}_ox(t).\]

Let us pause here for an illustrative example.

\textit{Example 3:} We adopt Example 3 from \cite{Fernando2025} to analyze a control input vector subject to a time-delay $\tau_u$, denoted as $u(t-\tau_u)$. The system matrices are given by

$A=\begin{pmatrix} 
	1 &0.5 &-1 &0 &1\\0.3 &0.5 &-0.6 &-0.3 &0.3\\-0.6 &0 &0.2 &0.6 &-0.6\\1.25 &0.5 &-1 &-0.25 &1.75\\-0.75 &0 &0 &0.75 &-0.25 \end{pmatrix}$,

$B=\begin{pmatrix}1 &-1\\1 &1\\0 &0\\1 &0\\0 &1 \end{pmatrix}$ and $F_o=\begin{pmatrix}0.5 &1 &-2 &0.5 &2.5 \end{pmatrix}$.

While the original $F_o$ violates condition (\ref{a25}), choosing $R=F_oA=\begin{pmatrix}0.75 &1 &-2 &0.25 &2.25\end{pmatrix}$ as suggested in \cite{Fernando2025} yields an augmented matrix $\bar{F}_o = \begin{pmatrix} F_o \\ R \end{pmatrix}$ that successfully satisfies the condition.

Thus, we obtain the following reduced-order system
\[\bar{z}_o(t)=\bar{N}_o\bar{z}_o(t)+\bar{B}_ou(t-\tau_u),\]
where $\bar{N}_o:=\bar{F}_oA\bar{F}^-_o=\begin{pmatrix}0 &1\\-0.5 &1.5\end{pmatrix}$ and $\bar{B}_o:=\bar{F}_oB=\begin{pmatrix}2 &3\\2 &2.5\end{pmatrix}$.

Now, for a given $\tau_u$, say, $\tau_u=0.5\text{s}$, the LMI in Lemma 11 \cite{trinhnam26} is feasible for $\lambda=1$, and we obtain
\[\bar{Z}_{\tau_u}=\begin{pmatrix} 2.8927 &-4.9739\\-2.1246 &3.2658 \end{pmatrix}.\]
Consequently, the target output controller 
\[u(t-0.5)=\bar{Z}_{\tau_u}\bar{z}_o(t-0.5),\]
guarantees the asymptotic stability of the following augmented target output vector
\begin{align}
	\dot{\bar{z}}_o(t)&=\begin{pmatrix}0 &1\\-0.5 &1.5\end{pmatrix}\bar{z}_o(t)\nonumber\\&+\begin{pmatrix}-0.5885 &-0.1504\\0.4738 &-1.7833 \end{pmatrix}\bar{z}_o(t-0.5).\nonumber\end{align}
For the above time-delay system, the eigenvalues located to the right of the vertical line $\alpha = -4$ are calculated as $\{-0.4537 \pm j1.5519, -0.4646, -3.6765\}$, using the approach described in \cite{wu}.

Now, by applying the target output controller to the full-order system (\ref{a1}), we obtain the closed-loop system
\begin{align}
	\label{a29}
	\dot{x}(t)=Ax(t)+B\bar{Z}_{\tau_u}\bar{F}_ox(t-0.5).\end{align}
The eigenvalues of the spectrum $\sigma\Big(A + B\bar{Z}_{\tau_u}\bar{F}_oe^{-\lambda \tau_u}\Big)$ lying to the right of the vertical line $\alpha = -4$ are computed as$$\{0.5, 0.2, -1, -0.4537 \pm j1.5519, -0.4646, -3.6765\},$$which contain the four stable eigenvalues of the reduced-order closed-loop system and the remaining three eigenvalues of $A$ at $\{0.5, 0.2, -1\}$.\}. Note that $\sigma(A)=\{-1, 1, 0.2, 0.5, 0.5\}$ and $\sigma(\bar{N}_o)=\{0.5, 1\}$. 

\begin{figure}[!h]
	\centering
	\includegraphics[width=0.9\linewidth]{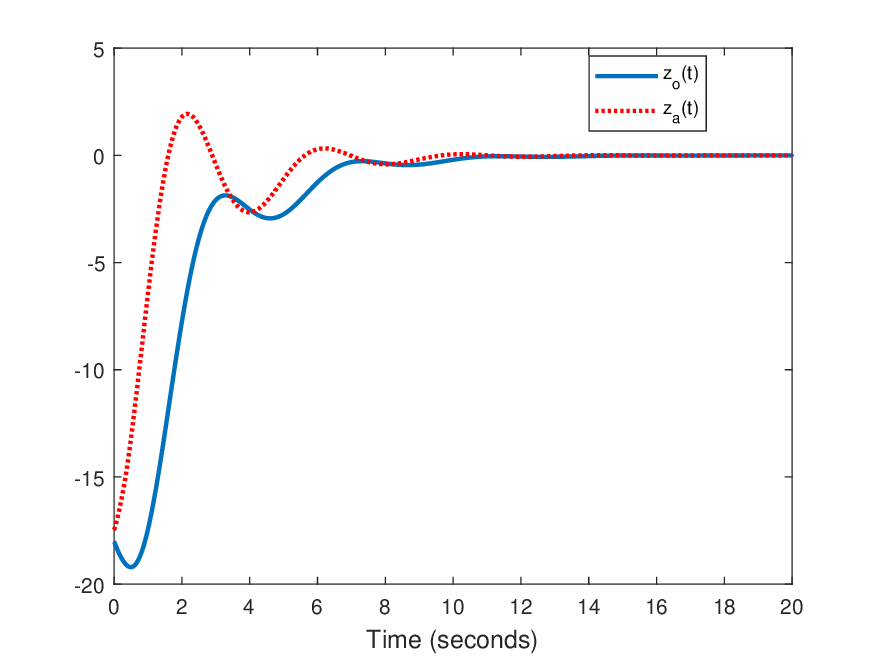}
	\caption{Trajectories of $z_o(t)$ and $z_a(t)$ under the proposed target output controller}
	\label{fig4paper3}
\end{figure}
We also simulate system (\ref{a29}) under the initial conditions $F_ox(0)=-18$, $Rx(0)=-17.5$. Figure~\ref{fig4paper3} displays the trajectories of $z_o(t)$ and $z_a(t)$, which clearly converge to $0$ as $t \rightarrow \infty$. Although system (\ref{a29}) is inherently unstable, the target output vector $\bar{z}_o(t)$ is successfully stabilized by the delayed target output controller $u(t-0.5)=\bar{Z}_{\tau_u}\bar{z}_o(t-0.5)$. This highlights a key advantage: regulating a specific target output is achievable even when forcing the entire state vector $x(t)$ to converge to zero is impossible due to the uncontrollability of the matrix pair $(A,B)$. Whereas, the triplet $(A,B,F_o)$ is target output
controllable \cite{Fernando2025}. Since target output controllability is a more relaxed condition than the controllability condition. This benefit previously reported in \cite{Fernando2025}.

Next, we consider the practical scenarios where neither $z_o(t)$ nor its delayed counterpart $z_o(t-\tau_u)$ is directly available for feedback. Additionally, the measured output vector is subject to the time delay modeled in (\ref{a2}). 

Based on (\ref{a23}), we consider two scenarios: Case 1, where $r \leq m$ and $Z_{\tau_u}$ is of full row rank, and Case 2, where $r > m$ and $Z_{\tau_u}$ is of full column rank. For Case 1, a functional observer is designed to directly estimate the delayed functional
\begin{align}
	\label{a30}
	z_r(t)=Z_{\tau_u}z_o(t-\tau_u)=Fx(t-\tau_u),
\end{align}
where $F:=Z_{\tau_u}F_o\in \mathbb{R}^{r\times n}$.

Consequently, the control law (\ref{a23}) is realized using the estimate $\hat{z}_r(t)$. For Case 2, the functional observer is instead designed to estimate the delayed functional
\begin{align}
	\label{a31}
	z_m(t)=F_ox(t-\tau_u),
\end{align}
allowing the controller (\ref{a23}) to be implemented via $Z_{\tau_u}\hat{z}_m(t)$. Section~\ref{asb} details the comprehensive design of the functional observers tailored to both scenarios.

To illustrate this, let us return to Example 3 and consider the case where the output vector is defined as in (\ref{a2}), with$$C=\begin{pmatrix}1.25 &2 &-4 &0.75 &4.75\\0 &0 &0 &0 &1\end{pmatrix}$$and $\tau_y=0.7\text{s}$. In this scenario, the time delay in the output vector exceeds the time delay in the input vector, resulting in mismatched time delays that occur simultaneously in both the control input and system output vectors. Note that the triplet $(A,C,\bar{F}_o)$ is functionally observable \cite{fern2}.

Thus, we design a functional observer to estimate the following functional
\[z_r(t)=\bar{F}x(t-0.5), \ \bar{F}:=\bar{Z}_{\tau_u}\bar{F}_o,\]
$$\bar{F}=\begin{pmatrix}-2.2841 &-2.0812 &4.1624 &0.2029 &-3.9596\\1.3870 &1.1412 &-2.2824 &-0.2459 &2.0365\end{pmatrix}$$ by utilizing the delayed output vector $y(t)=Cx(t-0.7)$.

For the detailed design methodology using observer (\ref{a11})-(\ref{a12}), the reader is referred to Scenario 2 of Section \ref{asb}. Following this procedure, we obtain the following functional observer
\begin{align}
	\hat{z}_r(t)&= w(t)+\begin{pmatrix}-1.2036 &1.0578\\0.7343 &-0.7072\end{pmatrix}y(t-0.5),\nonumber\\
	\dot{w}(t)&=\begin{pmatrix}1.4131 &0.7535\\-0.5007 &0.0869\end{pmatrix}w(t)\nonumber\\&+\begin{pmatrix}-1.4211 &  -0.2267\\
		0.3799  & -0.6415\end{pmatrix}w(t-0.7)\nonumber\\&+\begin{pmatrix}-1.2936 &   2.0197\\
		0.6935 &   -1.2983\end{pmatrix}y(t-0.5)\nonumber\\&+\begin{pmatrix}1.5441  & -1.3430\\
		-0.9282  &  0.8555\end{pmatrix}y(t-1.2) \nonumber\\&+\begin{pmatrix}-4.1624  & -3.7567\\
		2.2824 &    1.7906\end{pmatrix}u(t-1)   	\nonumber\\&+\begin{pmatrix}4.8146  &  5.5623\\
		-2.9370  & -3.3312\end{pmatrix}u(t-1.7). \nonumber  
\end{align}
Finally, the target output controller$$u(t-0.5)=\bar{Z}_{\tau_u}\bar{z}_o(t-0.5)$$is implemented using the functional observer designed above. To illustrate the validity of this observer-based control scheme, the trajectories of $z_o(t)$ and $z_a(t)$ are displayed in Figure \ref{fig5paper3}. The results clearly demonstrate that asymptotic stability of the target output $z_o(t)$ is achieved. Note that its convergence to zero is slower than in the case without an observer (see Figure \ref{fig4paper3}). This delay is due to the asymptotic stability of the observer error dynamics
\begin{align}\dot{e}(t) &= Ne(t) + N_{\tau_y}e(t-0.7)=\begin{pmatrix} 1.4131 & 0.7535 \\ -0.5007 & 0.0869 \end{pmatrix}e(t)\nonumber\\&+ \begin{pmatrix} -1.4211 & -0.2267 \\ 0.3799 & -0.6415 \end{pmatrix} e(t-0.7), \nonumber\end{align}
where the stable dominant eigenvalues (calculated using \cite{wu}) are located at $\{-0.2916 \pm j0.916\}$, resulting in a longer settling time.
\begin{figure}[!h]
	\centering
	\includegraphics[width=0.9\linewidth]{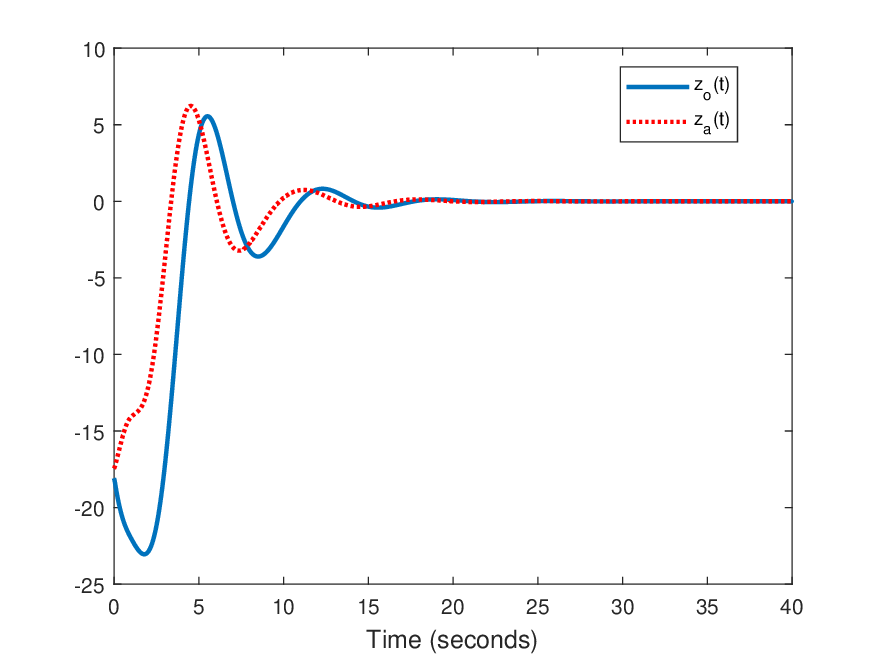}
	\caption{Trajectories of $z_o(t)$ and $z_a(t)$ under the proposed observer-based delayed target output controller}
	\label{fig5paper3}
\end{figure}

\subsection{Design of Delayed Functional Observers via Row-Space State Projection}\label{asd}
Section \ref{asc} presented the design of delayed functional observers for estimating the target output controllers (\ref{a30})-(\ref{a31}) using the full-order system (\ref{a1})-(\ref{a2}).

Alternatively, consider the reduced-order subsystem (\ref{a27}), constructed by projecting the full state dynamics onto the row space of $F_o$. To this end, we define the delayed output vector
\begin{align} \label{a32}
	y_o(t)=K_oy(t)=C_oz_o(t-\tau_y),
\end{align}
where the matrices $K_o\in\mathbb{R}^{l\times p}$ (with $0<l\leq p$) and $C_o\in\mathbb{R}^{l\times m}$ are to be determined. From (\ref{a32}), it follows that $K_o$ and $C_o$ must satisfy the algebraic constraint
\begin{align} \label{a33}
	K_oC=C_oF_o.
\end{align}
Equation (\ref{a33}) can be expressed as follows
\begin{align} 	\label{a34}\begin{pmatrix} K_o &C_o \end{pmatrix}\begin{pmatrix} C\\-F_o \end{pmatrix}=\bf 0.\end{align}
Let $v := \mathrm{rank} \begin{pmatrix} C^{\mathsf T} & -F_o^{\mathsf T} \end{pmatrix}^{\mathsf T}$. Assuming $v < p+m$, the matrix in (\ref{a34}) is not of full row rank, ensuring the existence of a non-trivial solution. The general solution to (\ref{a34}) can then be expressed as
\begin{align} 	\label{a35}\begin{pmatrix} K_o &C_o \end{pmatrix}=Z_o\left(I_{p+m}-\begin{pmatrix} C\\-F_o \end{pmatrix}\begin{pmatrix} C\\-F_o \end{pmatrix}^-\right),\end{align}
where $Z_o\in\mathbb{R}^{l\times (p+m)}$ is arbitrary.

Consequently, we obtain the following reduced-order system
\begin{align}
	\label{a36}
	\dot{z}_o(t)&= N_oz_o(t) +B_ou(t-\tau_u),\\
	\label{a37}y_o(t)&=C_oz_o(t-\tau_y).
\end{align}
To estimate $z_o(t)$ or its delayed counterpart from the reduced-order system, the matrix pair $(N_o,C_o)$ must be observable. This, in turn, requires the functional observability \cite{fern2} of the triplet $(A,K_oC,F_o)$. A logical prerequisite for this condition is that the reduced-order system restricts itself to a subset of $l$ outputs from the original $p$ outputs, defined by $y_o(t)=K_oy(t)$. Consequently, the functional observability of $(A,K_oC,F_o)$ serves as the necessary condition to successfully estimate $z_o(t)$ using this reduced set of outputs.

Based on (\ref{a36})-(\ref{a37}), we can design an observer to estimate either $z_o(t-\tau_u)$ or $u(t-\tau_u)$, thereby enabling the implementation of the output feedback controller. Because these observers are synthesized using the $m$-dimensional reduced-order system \eqref{a36}-\eqref{a37} and its lower-dimensional measurement vector $y_o(t)$, they possess a lower architectural complexity than those observers designed for system \eqref{a1}, which uses the entire output vector $y(t)$.

Let us return to Example 3 where we obtained
\[\bar{N}_o=\begin{pmatrix}0 &1\\-0.5 &1.5\end{pmatrix}, \bar{B}_o=\begin{pmatrix}2 &3\\2 &2.5\end{pmatrix}.\] With $C$ and $\bar{F}_o$ as given, the matrix $\begin{pmatrix} C\\-\bar{F}_o \end{pmatrix}$ is not full row rank. Thus, from (\ref{a35}) by substituting $\bar{F}_o$ for $F_o$, and by letting
\[Z_o=\begin{pmatrix} 3 &0 &0 &0 \end{pmatrix}\] we  obtain
\[K_o=\begin{pmatrix} 1 &0 \end{pmatrix}, \ C_o=\begin{pmatrix} 1 &1 \end{pmatrix}.\]
Note that the pair $(\bar{N}_o, C_o)$ is observable. Utilizing the low-dimensional matrices $\bar{N}_o$, $\bar{B}_o$, and $C_o$, we design an observer to estimate the target control law$$z_r(t)=u(t-0.5)=\bar{Z}_{\tau_u}\bar{z}_o(t-0.5),$$
where
$$\bar{Z}_{\tau_u}=\begin{pmatrix} 2.8927 &-4.9739\\-2.1246 &3.2658 \end{pmatrix}.$$
With $\tau_u=0.5\text{s}$, $\tau_y=0.7\text{s}$, we employ observer (\ref{a11})-(\ref{a12}) to estimate $z_r(t)$. Based on the systematic design procedure detailed in \cite{trinhnam26}, the LMI in Lemma 11 \cite{trinhnam26} is feasible for $\lambda=1$, we obtain the following observer
\begin{align}
	\hat{z}_r(t)&= w(t)+\begin{pmatrix}-3.3192\\
		2.1487\end{pmatrix}y_o(t-0.5),\nonumber\\
	\dot{w}(t)&=\begin{pmatrix}1.4131 &0.7535\\-0.5007 &0.0869\end{pmatrix}w(t)\nonumber\\&+\begin{pmatrix}-1.4211 &  -0.2267\\
		0.3799  & -0.6415\end{pmatrix}w(t-0.7)\nonumber\\&+\begin{pmatrix}-1.1019\\
		0.4612\end{pmatrix}y_o(t-0.5)+\begin{pmatrix}4.2300\\
		-2.6392\end{pmatrix}y_o(t-1.2) \nonumber\\&+\begin{pmatrix}-4.1624  & -3.7567\\
		2.2824 &    1.7906\end{pmatrix}u(t-1)   	\nonumber\\&+\begin{pmatrix}13.2768 &  18.2556\\
		-8.5947 & -11.8178\end{pmatrix}u(t-1.7). \nonumber  
\end{align}
The above observer utilizes $y_o(t)$ but not the entire vector $y(t)$ hence it achieves a lower structural complexity than an observer designed for the full-order system, which uses the entire output vector $y(t)$.

Conversely, it is worth noting that the triplet $(A,C_1,\bar{F}_o)$--where $C_1$ denotes the first row of $C$--is functionally observable. Synthesizing an observer \eqref{a11}-\eqref{a12} based on $A$, $C_1$, $B$, and $\bar{F}_o$ yields an estimator identical to the aforementioned observer (with the LMI in Lemma 11 \cite{trinhnam26} being feasible for $\lambda=1$). Consequently, both approaches offer equivalent and viable alternatives for designing delayed functional observers.

\textit{Remark 6:} We address the case where, for a given $F_o$, the matrix$$\begin{pmatrix} C \\ -F_o \end{pmatrix}$$has full row rank as well as condition (\ref{a25}) is violated. In this situation, we enlarge the target output vector by augmenting $F_o$ with selected rows (or linear combinations of rows) of $C$ along with their corresponding observability subspaces. Let $C_s$ denote a submatrix consisting of one or more rows (or linear combinations of rows) of $C$. We define the augmented observability matrix as
\begin{equation}	\mathcal{O}_{(A,F_o,C_s)}:=	\begin{pmatrix}
		F_o \\ F_oA \\ \vdots \\ F_oA^{n-1} \\
		C_s \\ C_sA \\ \vdots \\ C_s A^{n-1}
	\end{pmatrix},
\end{equation}
and its rank as
$$q := \mathrm{rank}\left(\mathcal{O}_{(A,F_o,C_s)}\right).$$
For the case where no row of $C$ lies within the row space of $F_o$, the rows of $C_s$ are linearly independent of $F_o$. Consequently, we can construct a new full-row-rank matrix, denoted as $\tilde{F}_o$, by selecting $q$ linearly independent rows from $\mathcal{O}_{(A,F_o,C_s)}$ such that the rows of both $F_o$ and $C_s$ are preserved, i.e.,
\begin{equation*}
	\tilde{F}_o =
	\begin{pmatrix} F_o \\ C_s \\ R \end{pmatrix}
	\in \mathbb{R}^{q\times n},
\end{equation*}
where $R$ is the $(q - m - \mathrm{rank}(C_s))\times n$ matrix created by basis rows of $ \{F_oA, F_oA^2, \ldots, F_oA^{n-1},\allowbreak C_sA, C_s A^2, \ldots\, C_sA^{n-1}\}.$
By construction, $\mathrm{row}(\tilde{F}_o) = \mathrm{row}(\mathcal{O}_{(A,F_o,C_s)})$.

	Since $\mathrm{rank}(\mathcal{O}_{(A,F_o,C_s)})=q$, we can select $q$
linearly independent rows of $\mathcal{O}_{(A,F_o,C_s)}$ that contain all the rows of
$F_o$. Denote the resulting full-row-rank matrix as
\[
\tilde{F}_o =
\begin{pmatrix} F \\ \tilde{R} \end{pmatrix} \in \mathbb{R}^{q\times n},
\qquad q \geq m,
\]
where $\tilde{R}\in\mathbb{R}^{q-m)\times n}$ collects the
additional rows. By construction,
$\mathrm{row}(\tilde{F}_o) = \mathrm{row}(\mathcal{O}_{(A,F_o,C_s)})$.

By the Cayley--Hamilton theorem,
both $F_oA^n$ and $C_sA^n$ are linear combinations of
$F_o, F_oA, \ldots, F_oA^{n-1}$ and $C_s, C_sA, \ldots, C_sA^{n-1}$, respectively.
Hence
\begin{equation*}
	\mathrm{row}(\tilde{F}_oA) \subseteq \mathrm{row}(\mathcal{O}_{(A,F_o,C_s)})
	= \mathrm{row}(\tilde{F}_o),
\end{equation*}
and therefore
\begin{equation}
	\label{rank}\mathrm{rank}\!\begin{pmatrix} \tilde{F}_oA \\ \tilde{F}_o \end{pmatrix}
	= \mathrm{rank}(\tilde{F}_o).
\end{equation}
By Lemma~\ref{lemma2p3}, there exists a matrix
$\tilde{N}_o:= \tilde{F}_oA\tilde{F}_o^- \in \mathbb{R}^{q\times q}$
satisfying $\tilde{N}_o \tilde{F}_o=\tilde{F}_oA$.

We can now find $K_o\in\mathbb{R}^{l\times p}$ and $C_o\in\mathbb{R}^{l\times q}$ via the following formula
\begin{equation}\label{a39}
	\begin{pmatrix} K_o & C_o \end{pmatrix}
	= \tilde{Z}\left( I_{p+q} -
	\begin{pmatrix} C \\ -\tilde{F}_o \end{pmatrix}
	\begin{pmatrix} C \\ -\tilde{F}_o \end{pmatrix}^{-} \right),
\end{equation}
where $\tilde{Z} \in \mathbb{R}^{l\times(p+q)}$ is arbitrary, and
$l = p + q - \mathrm{rank}\!\begin{pmatrix} C \\ \tilde{F}_o\end{pmatrix}>0$.
Finally, we obtain the following reduced-order system
\begin{align}
	\label{a40}
	\dot{\tilde{z}}_o(t) &= \tilde{N}_o\tilde{z}(t) + \tilde{B}_o u(t-\tau_u),\\
	\label{a41}
	y_r(t) &= K_o y(t) = C_o\,\tilde{z}_o(t - \tau_y),
\end{align}
where $\tilde{B}_o:=\tilde{F}_oB$, $\tilde{z}_o(t) := \tilde{F}_ox(t) =
\begin{pmatrix} z_o(t) \\ \tilde{z}_a(t) \end{pmatrix}$,
$\tilde{z}_a(t) := \begin{pmatrix}
	C_s\\R
\end{pmatrix}x(t)$.

\textit{Example 4:}  Let $A=\begin{pmatrix} 0  &   1 &    1\\
	0  &   -1  &   2\\
	0  &   -3  &   2 \end{pmatrix}$, $B=\begin{pmatrix}1\\2\\3 \end{pmatrix}$, $C=\begin{pmatrix} 0  &  0 &    1\\
	1  &  0 &0 \end{pmatrix}$, and $F_o=\begin{pmatrix} 0  & 1  & 1 \end{pmatrix}$.

In this example, $F_o$ does not contains any row or combinations of $C$, i.e., $\mathrm{rank}\begin{pmatrix}F_o\\ C\end{pmatrix}=3$. Furthermore, $\mathrm{rank}\begin{pmatrix}F_oA\\ F_o\end{pmatrix}\neq \mathrm{rank}\begin{pmatrix}F_o\end{pmatrix}$. So as discussed in Remark 6, we enlarge the target output vector. 

First, let us pick $C_s= C_1=\begin{pmatrix}
	0&0&1
\end{pmatrix}$.
We note that, the matrix
\begin{equation*}
	\mathcal{O}_{(A,F_o,C_1)} = \begin{pmatrix}
		0& 1& 1\\
		0& -4& 4\\
		0& -8& 0\\
		0& 0& 1\\
		0& -3& 2\\
		0& -3& -2
	\end{pmatrix}
\end{equation*}	
has $q = \mathrm{rank}(\mathcal{O}_{(A,F_o,C_1)}) = 2$. Hence, we can choose 
\begin{eqnarray*}
	\tilde{F}_o = \begin{pmatrix} F_o \\ C_1 \end{pmatrix}
	= \begin{pmatrix}
		0&1 &1 \\
		0 & 0 & 1	\end{pmatrix}.
\end{eqnarray*}
It is easy to check that the rank condition (\ref{rank}) is satisfied, and also the triplet $(A,C_1,F_o)$ is functionally observable. We obtain 
\begin{eqnarray*}
	\tilde{N}_o = \tilde{F}_oA\tilde{F}_o^- =
	\begin{pmatrix}
		 -4 & 8 \\
		 -3 & 5 
	\end{pmatrix},
	\
	\tilde{B}_o = \tilde{F}_oB = \begin{pmatrix} 5 \\ 3 \end{pmatrix}.
\end{eqnarray*}
We note that $l = p + q - \mathrm{rank}\!\begin{pmatrix} C \\ \tilde{F}_o\end{pmatrix} = 1.$ By choosing $Z = \begin{pmatrix} 1 & 1 & 1 & 1 \end{pmatrix}$ and solving \eqref{a39}, we obtain $K_0 = \begin{pmatrix}
	1&0
\end{pmatrix}$, and
$C_o = \begin{pmatrix} 0 & 1 \end{pmatrix}$.

The resulting second-order system is
\begin{align*}
	\dot{\tilde{z}}_o(t) &=
	\begin{pmatrix}
		-4 & 8 \\
		-3 & 5 
	\end{pmatrix}\tilde{z}_o(t)
	+ \begin{pmatrix} 5 \\ 3 \end{pmatrix} u(t-\tau_u), \\
	y_o(t) &= y_1(t) =
	\begin{pmatrix} 0 & 1 \end{pmatrix} \tilde{z}_o(t-\tau_y).
\end{align*}
Note that the pair $(\tilde{N}_o, C_o)$ is observable (the triplet $(A,C_1,F_o)$ is functionally observable). Finally, selecting $C_s=C_2$ recovers all measurements but yields a third-order system. Because this system is equivalent in dimension to the original state space, it provides no reduction in order.

Based on the reduced-order system derived above, we now design an observer-based controller to drive the target output vector to the origin ($z_o(t) \to \mathbf{0}$ as $t \to \infty$) from an arbitrary initial condition $F_ox(0)$. For illustrative purposes, let the time delays be given as $\tau_u = 0.5\text{s}$ and $\tau_y = 0.7\text{s}$. Because $\tau_y > \tau_u$, this scenario effectively represents a system with mismatched time delays occurring simultaneously in both the control input and system output vectors.

Note that $\tilde{N}_o$ is not Hurwitz. With $\tau_u=0.5\text{s}$, the LMI in Lemma 11 \cite{trinhnam26} is feasible for $\lambda=1$, and we obtain the following target output controller 
\[u(t-0.5)=\tilde{Z}_{\tau_u}\tilde{z}_o(t-0.5),\]
where

\[\tilde{Z}_{\tau_u}=\begin{pmatrix} 0.4161   &-0.9443 \end{pmatrix}.\]
Consequently, the target output controller guarantees the asymptotic stability of the following augmented target output vector
\begin{align}
	\dot{\tilde{z}}_o(t)&=\begin{pmatrix}
		-4 & 8 \\
		-3 & 5 
	\end{pmatrix}\tilde{z}_o(t)+\begin{pmatrix}2.0806 &  -4.7213\\
	1.2484 &  -2.8328 \end{pmatrix}\tilde{z}_o(t-0.5).\nonumber\end{align}
For the above time-delay system, the eigenvalues located to the right of the vertical line $\alpha = -4$ are calculated as $\{-0.3668 \pm j2.1797, -1.6662\}$, using the approach described in \cite{wu}.

Now, by applying the same target output controller to the full-order system (\ref{a1}), we obtain the closed-loop system
\begin{align}
	\dot{x}(t)=Ax(t)+B\tilde{Z}_{\tau_u}\tilde{F}_ox(t-0.5).\nonumber \end{align}
The eigenvalues of the spectrum $\sigma\Big(A + B\bar{Z}_{\tau_u}\bar{F}_oe^{-\lambda \tau_u}\Big)$ lying to the right of the vertical line $\alpha = -4$ are computed as$$\{0, -0.3668 \pm j2.1797, -1.6662\},$$which contain the three stable eigenvalues of the reduced-order closed-loop system. Figure~\ref{fig6paper3} displays the trajectories of $z_o(t)$ which clearly converge to $0$ as $t \rightarrow \infty$. 
\begin{figure}[!h]
	\centering
	\includegraphics[width=0.9\linewidth]{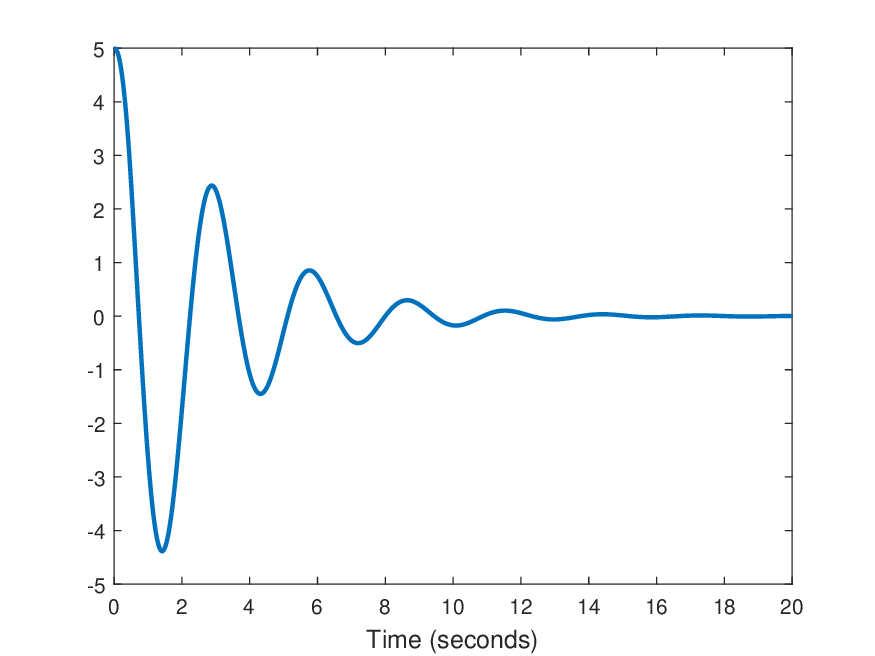}
	\caption{Trajectories of $z_o(t)$ under the proposed target output controller}
	\label{fig6paper3}
\end{figure}

Next, based on the derived second-order system, we design a functional observer to estimate the control input functional defined as$$z(t) := u(t-0.5) = \tilde{Z}_{\tau_u}\tilde{z}_o(t-0.5),$$where the observer utilizes only one delayed output measurement$$y_o(t) = y_1(t) = \begin{pmatrix} 0 & 1 \end{pmatrix} \tilde{z}_o(t-0.7).$$

As discussed in Scenario 2 of Section \ref{asb}, we employ observer (\ref{a11})-(\ref{a12}) to estimate $z(t)$, and we obtain the following second-order observer
\begin{align}
	\hat{z}(t)&= w_1(t)-0.0409y_1(t-0.5),\nonumber\\
	\dot{w}(t)&=\begin{pmatrix} 2.8076  &  1.2588\\-7.2094 &   -1.8076\end{pmatrix}w(t)\nonumber\\&+\begin{pmatrix}-0.2946 &  -0.2968\\
		0.9337  & -0.2018\end{pmatrix}w(t-0.7)\nonumber\\&+\begin{pmatrix}0.2712\\
		0.4964\end{pmatrix}y_1(t-0.5)+\begin{pmatrix}-0.0264\\
		-0.0643\end{pmatrix}y_1(t-1.2) \nonumber\\&+\begin{pmatrix}-0.7522\\
		3\end{pmatrix}u(t-1)+\begin{pmatrix}0.1226\\
		-0.3885\end{pmatrix}u(t-1.7). \nonumber  
\end{align}
Finally, the target output controller$$u(t-0.5)=\tilde{Z}_{\tau_u}\tilde{z}_o(t-0.5)$$is implemented using the functional observer designed above. Note that matrix $M=\begin{pmatrix}-0.0409\\0.1295
\end{pmatrix}$. To illustrate the validity of this observer-based control scheme, the trajectories of $z_o(t)$ are displayed in Figure \ref{fig7paper3}. The results clearly demonstrate that asymptotic stability of the target output $z_o(t)$ is achieved.
\begin{figure}[!h]
	\centering
	\includegraphics[width=0.9\linewidth]{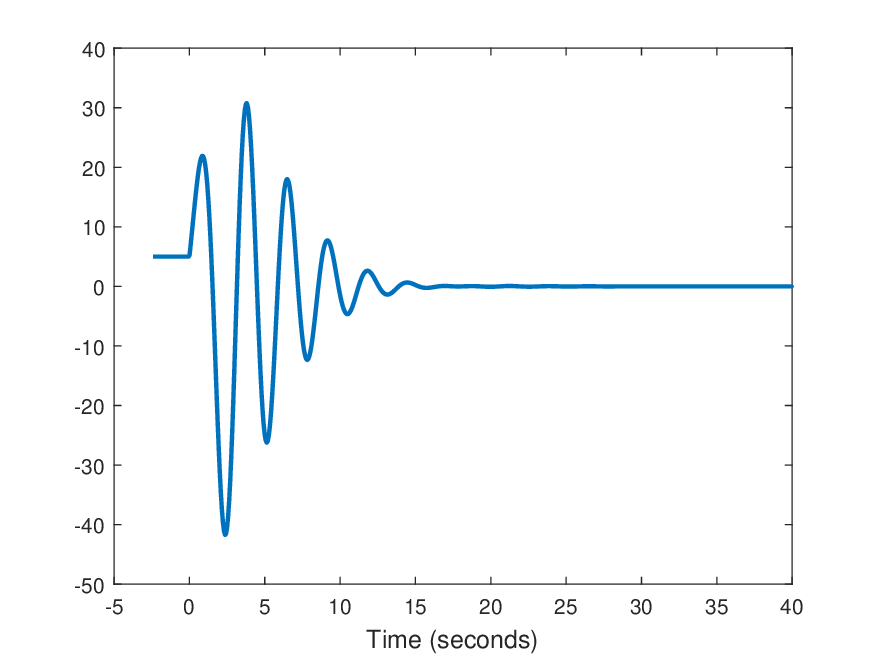}
	\caption{Trajectories of $z_o(t)$ under the proposed observer-based target output controller}
	\label{fig7paper3}
\end{figure}

\textit{Example 5:}  We adopt Example 4 from \cite{Fernando2025} to design an observer-based target controller for the system where it is subject to mismatched time delays in both the state and output vectors, where $\tau_u=1.7\text{s}$ and $\tau_y=3\text{s}$. The system matrices are given by

$A=\begin{pmatrix} -0.5 &0.5 & -1 &-0.5 &0.5\\-0.7 &-0.5 &1.4 &0.7 &-0.7\\-0.6 &0 &0.2 &0.6 &-0.6\\0.25 &0.5 &-1 &-1.25 &0.75\\-0.25 &0 &0 &0.25 &-0.75\end{pmatrix}$,

$B=\begin{pmatrix}1 &-1\\2 &1\\0.5 &1\\1 &-1\\0 &2 \end{pmatrix}$, $F_o^{\mathsf T}=C^{\mathsf T}=\begin{pmatrix} 0.5 &-0.5\\0 &0\\0 &2\\-0.5 &0.5\\0.5 &-0.5\end{pmatrix}$.

While the system is unstable, uncontrollable, and unobservable, the triplet $(A,B,F_o)$ is target output controllable \cite{Fernando2025} and $(A,C,F_o)$ is functionally observable. Thus, a static output feedback controller can be designed by placing two poles of a second-order subsystem anywhere in the complex plane, driving the target output vector to zero ($z_o(t) \to \mathbf{0}$) from any initial state $F_ox(0)$ \cite{Fernando2025}.

Now, consider mismatched state and output time delays where $\tau_u=1.7\text{s}$ and $\tau_y=3\text{s}$ ($\tau_y>\tau_u$). Because of the heavy measurement delay, $z_o(t-\tau_u)$ is unavailable, and only $z_o(t-\tau_y)$ is known. Therefore, $y(t)$ must be used to reconstruct a less-delayed functional estimate of $z_o(t-\tau_u)$ to implement the control law $u(t-\tau_u)$.

First, given $F_o$ and $A$, condition (\ref{a25}) is satisfied. Furthermore, since $F_o=C$, we obtain the following second-order time-delay system
\begin{align}
	\dot{z}_o(t) &= N_oz_o(t) + B_o u(t-1.7),\nonumber\\
y(t) &= C_o\,z_o(t - 3),\nonumber
\end{align}
where 

$N_o:=F_oAF_o^-=\begin{pmatrix} -1 &0\\-1.2 &0.2 \end{pmatrix}$, $B_o:=F_oB=\begin{pmatrix} 0 &1\\1 &1 \end{pmatrix}$ and  $C_o=I_2$.

Let us now design a target control law 
\[u(t-1.7)=Z_{\tau_u}z_o(t-1.7)\]
to stabilize the above time-delay system.  With $\tau_u=1.7\text{s}$, the LMI in Lemma 11 \cite{trinhnam26} is feasible for $\lambda=1$, and we obtain
\[Z_{\tau_u}=\begin{pmatrix} 0.5718 &  -0.3875\\-0.2346 &0.0891 \end{pmatrix}.\]

Applying the control law $u(t-1.7)=Z_{\tau_u}z_o(t-1.7)=Z_{\tau_u}F_ox(t-1.7)$ to the full-order system $\dot{x}(t)=Ax(t)+Bu(t-1.7)$, Figure~\ref{fig8paper3} displays the trajectories of the vector $z_o(t)$, which clearly converges to $0$ as $t \rightarrow \infty$. 
\begin{figure}[!h]
	\centering
	\includegraphics[width=0.9\linewidth]{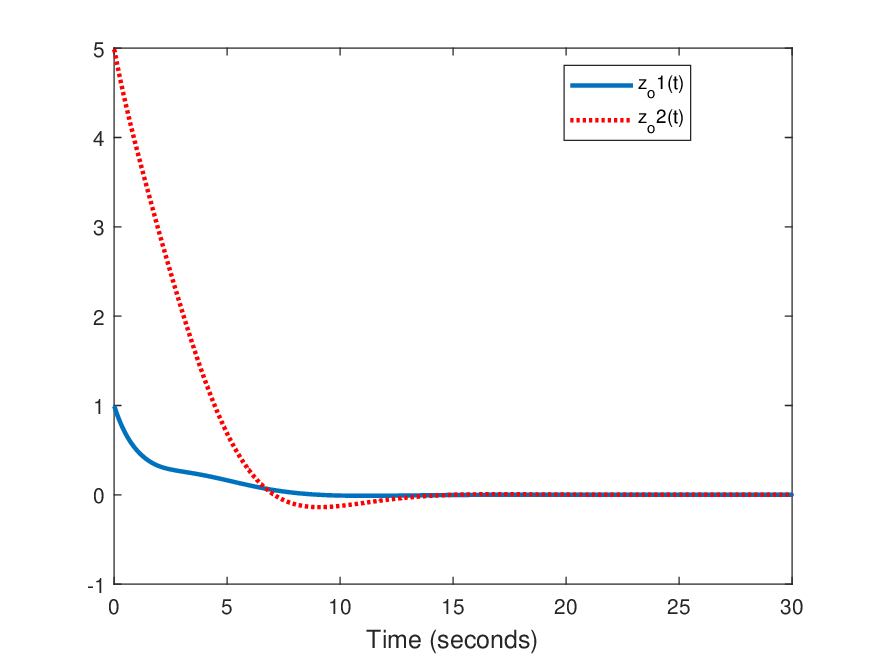}
	\caption{Trajectories of vector $z_o(t)$ under the proposed target output controller}
	\label{fig8paper3}
\end{figure}

Although $F_o=C$, the heavy measurement delay ($\tau_y > \tau_u$) prevents the direct implementation of the designed control law. This limitation arises because the required state vector $z_o(t-1.7)$ is unavailable, and only $z_o(t-3)$ can be accessed. To overcome this, we can leverage the framework developed in Scenario 2 (Section \ref{asb}) and utilize either observer (\ref{a11})-(\ref{a12}) or (\ref{a15}) (since $C_o= I_2$) to estimate the following functional which is the target control law
\[z(t)=u(t-1.7)=Z_{\tau_u}z_o(t-1.7).\]

Owing to the utilization of multiple internal delays, observer (\ref{a15}) exhibits a faster convergence rate to zero than observer (\ref{a11})-(\ref{a12}). For a detailed analysis, the reader is referred to Section V of \cite{trinhnam26}.

Thus, let us now employ observer (\ref{a15}), and we obtain the following observer (where the LMI in Lemma 13 in \cite{trinhnam26} is feasible for $\tau_u=1.7\text{s}$, $\tau_y=3\text{s}$ and $\lambda=1$)
\begin{align}
	\dot{\hat{z}}(t)&=\begin{pmatrix}0.6925   & 2.1428\\-0.3891 &-1.4925 \end{pmatrix}\hat{z}(t)\nonumber\\&+\begin{pmatrix} -0.5903  & -0.2003\\0.1539 &0.0018\end{pmatrix}\hat{z}(t-1.7) \nonumber\\&+ \begin{pmatrix} 0.1708    &0.0238\\-0.0595  & -0.0106 \end{pmatrix}\hat{z}(t-3)\nonumber\\&+\begin{pmatrix} 0.2905  & -0.2109\\-0.0876 &   0.0595 \end{pmatrix}y(t-0.4)\nonumber\\&+\begin{pmatrix} -0.0921  &  0.0641\\0.0315 &  -0.0221\end{pmatrix}y(t-1.7)\nonumber\\&+\begin{pmatrix} -0.3875  &  0.1843\\0.0891 &	-0.1456\end{pmatrix}u(t-3.4).\nonumber
\end{align}
To demonstrate the validity of this observer-based control scheme, the trajectories of $z_o(t)$ are displayed in Figure \ref{fig9paper3}. The results clearly show that the target output $z_o(t)$ achieves asymptotic stability.
\begin{figure}[!h]
	\centering
	\includegraphics[width=0.9\linewidth]{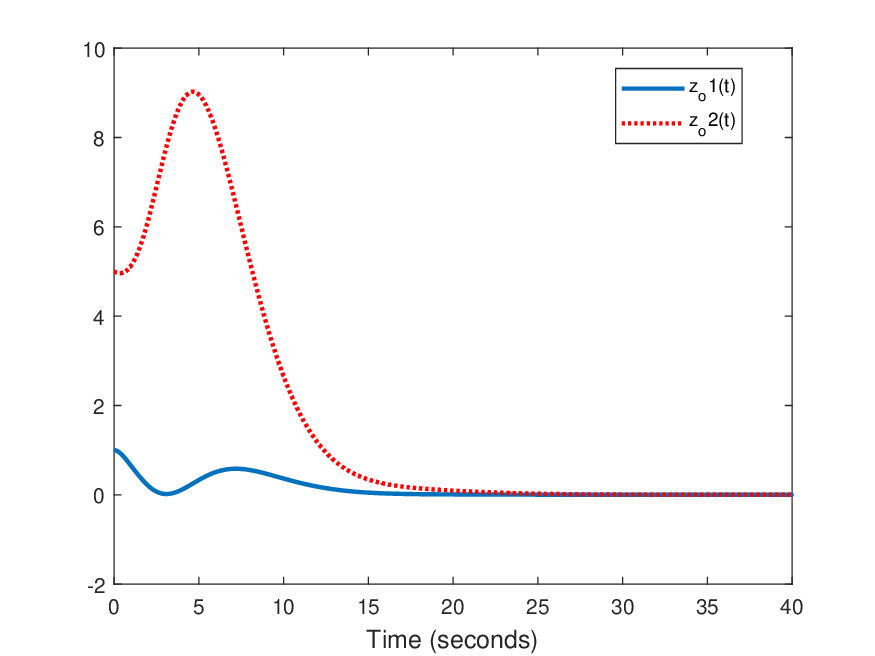}
	\caption{Trajectories of $z_o(t)$ under the proposed observer-based target output controller}
	\label{fig9paper3}
\end{figure}
 
\section{Conclusion}
This paper investigated the stabilization of linear systems subject to simultaneous, mismatched input and output time delays. First, an asymptotically stabilizing delayed state-feedback controller was synthesized using advanced LMI techniques. Second, this controller was realized via novel time-delay compensators \cite{trinhnam26}. This architecture accommodated an output delay $\tau_y$ independent of the input delay $\tau_u$, thereby enabling direct estimation of the delayed control law without restrictive matching conditions. Furthermore, this methodology was also extended to target output controllers \cite{Fernando2025} to account for simultaneous, mismatched time delays in both the control input and system output vectors. Finally, an alternative delayed functional observer design for  estimating target control law was introduced. It uses a reduced-order subsystem projected onto the row space of $F_o$, where the subsystem output acts as a subset of the full output vector.

\end{document}